\documentclass[a4paper,11pt,twoside]{article}

\usepackage{float}
\usepackage[utf8]{inputenc}
\usepackage[T1]{fontenc}
\usepackage[english]{babel}

\usepackage{charter}

\usepackage{xcolor}
\usepackage{verbatim}
\usepackage{listings}
\usepackage{multirow}

\usepackage[top=2.5cm, bottom=4.0cm, left=2.1cm, right=2.1cm]{geometry}
\usepackage{amsmath}
\usepackage{amsthm}	
\usepackage{amsfonts}	
\usepackage{amssymb	
,bbm
,units 
}
\usepackage{enumerate}
\usepackage[nobysame
,alphabetic
,initials
]{amsrefs}

\usepackage{graphicx}
\usepackage{epstopdf}
\usepackage{placeins}

\usepackage{
ulem		
,cancel		
} 

\numberwithin{equation}{section}
\usepackage[nodayofweek]{datetime}

\renewcommand*{\thefootnote}{\fnsymbol{footnote}}
\title{
Robust and Consistent Estimation of Generators in Credit Risk
}
\author{
\normalsize Greig Smith\footnote{G.~Smith was supported by The Maxwell Institute Graduate School in Analysis and its
Applications, a Centre for Doctoral Training funded by the UK Engineering and Physical
Sciences Research Council (grant [EP/L016508/01]), the Scottish Funding Council, Heriot-Watt
University and the University of Edinburgh.} \\[8pt]
\small  University of Edinburgh\\
\small Maxwell Institute for Mathematical Sciences\\ 
\small  School of Mathematics \\
\small  Edinburgh, EH9 3FD, UK\\  
\\
\\
\small G.Smith-13@sms.ed.ac.uk
\and
\normalsize Gon\c calo dos Reis\footnote{G.~dos Reis acknowledges support from the \emph{Funda{\c c}$\tilde{\text{a}}$o para a Ci$\hat{e}$ncia e a Tecnologia} (Portuguese Foundation for Science and Technology) through the project [UID/MAT/00297/2013] (Centro de Matem\'atica e Aplica\c c$\tilde{\text{o}}$es CMA/FCT/UNL).} \\[8pt]
\small  University of Edinburgh\\ 
\small  School of Mathematics \\
\small  Edinburgh, EH9 3FD, UK\\  
\small  and \\
\small  Centro de Matem\'atica e Aplica\c c$\tilde{\text{o}}$es \\
\small (CMA), FCT, UNL, Portugal \\
\small  G.dosReis@ed.ac.uk
}

\date{ \currenttime, \ddmmyyyydate\today}

\theoremstyle{plain}
\newtheorem{theorem}{Theorem}[section]
\newtheorem{lemma}[theorem]{Lemma}
\newtheorem{proposition}[theorem]{Proposition}

\newtheorem{definition}[theorem]{Definition}

\newtheorem{remark}[theorem]{Remark}

\newtheorem{assumption}[theorem]{Assumption}

\numberwithin{equation}{section}
\numberwithin{figure}{section}
\numberwithin{table}{section}

\newcommand{\bE}{\mathbb{E}}

\newcommand{\bN}{\mathbb{N}}
\newcommand{\bP}{\mathbb{P}}

\newcommand{\bR}{\mathbb{R}}

\newcommand{\cQ}{\mathcal{Q}}

\newcommand{\cX}{\mathcal{X}}
\newcommand{\cY}{\mathcal{Y}}

\definecolor{darkgreen}{rgb}{0,0.35,0}

\newcommand{\1}{\mathbbm{1}}

\theoremstyle{definition}

\begin{document}
\selectlanguage{english}

\title{Robust and consistent estimation of generators in credit risk}

%

\maketitle

\renewcommand*{\thefootnote}{\arabic{footnote}}

\begin{abstract} 

Bond rating Transition Probability Matrices (TPMs) are built over a one-year time-frame and for many practical purposes, like the assessment of risk in portfolios or the computation of banking Capital Requirements (e.g.~the new IFRS 9 regulation), one needs to compute the TPM and probabilities of default over a smaller time interval. In the context of continuous time Markov chains (CTMC) several deterministic and statistical algorithms have been proposed to estimate the generator matrix. We focus on the Expectation-Maximization (EM) algorithm by \cite{BladtSorensen2005} for a CTMC with an absorbing state for such estimation. 

This work's contribution is threefold. Firstly, we provide directly computable closed form expressions for quantities appearing in the EM algorithm and associated information matrix, allowing to easily approximate confidence intervals. Previously, these quantities had to be estimated numerically and considerable computational speedups have been gained. Secondly, we prove convergence to a single set of parameters under very weak conditions (for the TPM problem). Finally, we provide a numerical benchmark of our results against other known algorithms, in particular, on several problems related to credit risk. The EM algorithm we propose, padded with the new formulas (and error criteria), outperforms other known algorithms in several metrics, in particular, with much less overestimation of probabilities of default in higher ratings than other statistical algorithms.
\end{abstract}





\bigskip


\noindent{\bf Final author version.} Article to appear in \emph{Quantitative Finance} (submitted 2016 Dec 13; accepted 2017 Sep 15), {doi:10.1080/14697688.2017.1383627}  
\medskip

\noindent{\bf CRAN R-package} \emph{ctmcd: Estimating the Parameters of a Continuous-Time Markov Chain from Discrete-Time Data} 
-- ({https://CRAN.R-project.org/package=ctmcd}

\newpage

\section{Introduction}

Credit ratings play a key role not just in the calculation of a bank's capital charge (amount of capital a bank must hold) but also are typically a requirement for corporations wishing to issue bonds. There are different agencies which provide firms with a rating and the credit rating the agency gives a company determines in some respect the financial health of the company. Typically ratings are of the form $AAA$, $AA$, $A$, $BBB$, $BB$, $B$, $C$, $D$ (although it varies between agencies) with, $AAA$ the best (safest), $C$ the worst (riskiest) and $D$ to imply the firm has defaulted. It is also standard for banks to use their own internal ratings system (see \cite{YavinEtAl2014}). For an overview of the `science' involved in the rating procedure see \cite{Cantor2004}.

The main object of interest in this work is the so-called annual \emph{Transition Probability Matrix} (TPM), it is a stochastic matrix which shows the migration probabilities of different rated companies within a year. Rating agencies produce these annually. It is possible that such matrices are not initially stochastic due to company mergers or rounding for example. However, they can be renormalized by methods as described in \cite{KreininSidelnikova2001} and, as argued in \cite{BangiaDieboldKronimusEtAl2002}, renormalizing non rated companies across all ratings is indeed the industry standard. 

The main problem considered here is that a TPM encases transition probabilities over a $1$-year time frame and often in practice one needs a $3$ month or $10$ day transition matrix for which probabilities of default are lower than those in the TPM. Therefore one wants to accurately estimate the sub-annual matrix given the annual matrix. In the Basel proposals, Basel 3 \cite{Basel2013}*{p.3} a large part of the risk charge will be measured using ES (expected shortfall), which (as shown in \cite{ContEtAl2010}) is extremely sensitive to a small shifts in probabilities. Therefore, accurate and consistent estimation is essential in the calculation. Moreover, with the perspective of the \emph{IFTR 9} regulation, one needs to better estimate the related probabilities of default (PD) and Markov chain's generator since for a company whose risk profile changes significantly one needs to assess its risk throughout the bond's lifetime. As PD are given by exponential functions, smaller initial errors will be compounded in a significant way, especially in long termed bonds. This is compounded by the known fact, corroborated by our numerical experiments, that certain algorithms overestimate the PD. Our methodology yields a way to obtain point-in-time (PIT) estimates of the Probability of Default (PD) from the through-the-cycle (TTC) estimates. 

Credit rating models within the Markovian framework are handy both from a theoretical and numerical perspective. Evidence is given in \cite{LandoSkodeberg2002} that such Markovian structure is not true in practice, nonetheless, within the Markovian structure, efficient implementation of apt Markovian credit risk models and related risk measuring estimations able to deal with massive portfolios is a challenging problem, see \cites{BrigoMaiScherer2014, RutkowskiTarca2014}. There have been several models that produce non-Markovian effects such as mixing two generators \cite{FrydmanSchuermann2008} or considering hidden Markov models \cite{Korolkiewicz2012}, see also \cite{LongKeenanNeaguEtAl2011}. All non-Markovian models require in one way or another access to additional data for accurate calibration. This data is costly, needs to be updated over time and many companies opt to deal only with the TPMs. This work focuses on practitioners that do not have access to the data. The issue of rating momentum will be dealt with in forthcoming research. 
\smallskip

\textbf{The problem at hand.} 

We take the view of a financial agent who wishes to estimate probabilities of default or assess risk in his portfolio due to credit transitions but does not have access to (the expensive) individual credit rating transitions. The agent only has the annual TPM, say $\mathbf{P}(1)$, and uses a continuous time Markov chain (CTMC), say $(\hat{\mathbf{P}}(t))_{t\geq 0}$, with a finite state space to model the changes in rating over time. Under standard conditions the evolution of the CTMC can be written as $\hat{\mathbf{P}}(t)=e^{\mathbf{Q}t}$ where $\mathbf{Q}$ is the generator matrix. The problem is then to estimate $\mathbf{Q}$ given $\mathbf{P}(1)$. 
%
%
%
%
%
%
%
This estimation is non-trivial due to the so-called embeddability problem (not reviewed here). It is discussed in great detail by \cite{IsraelRosenthalWei2001} and, for more of the mathematics and many of the existing results on the embeddability problem, we point the reader to \cite{Lin2011}.  
%
%
%

Several approaches exist to tackle this estimation problem \cites{KreininSidelnikova2001, IsraelRosenthalWei2001,  TrueckOezturkmen2004, BladtSorensen2005, Inamura2006, BladtSorensen2009}, either using deterministic algorithms (e.g. diagonal or weighted adjustment, Quasi-optimization of the generator) or statistical ones (Expectation-Maximization (EM), Markov chain Monte-Carlo (MCMC) ones), see Section \ref{sec:CompetitorAlgorithms}. We focus on the Expectation-Maximization algorithm of \cite{BladtSorensen2005} for CTMCs and allow for an absorbing states.
\smallskip

\textbf{Contribution of this work.}
\begin{enumerate}

\item We provide sufficient conditions to extend the convergence result of \cite{BladtSorensen2005} to individual parameters rather than just convergence of likelihoods. The conditions presented are trivially satisfied in the context of the TPM problem. 

	\item We derive closed form expressions for the entries of the Hessian of the likelihood function used in the EM algorithm. This eliminated several instability issues appearing in other numerical implementations found in the literature and allows for computational speedups (comparatively). Moreover, the result provides a way to estimate the error of the estimation and assess the nature of the stationary point the algorithm has converged to.


	\item We give a short overview of known methods and implement them with some modifications as to improve their performance.  See Sections \ref{sec:CompetitorAlgorithms} \& \ref{Sec:Comparing Algorithms} for precise meanings: we apply the algorithms to certain credit risk problems and carry out a simulation study to check the impact in the computation of \emph{risk charges}, namely IRC (Incremental Risk Charge) with VaR (Value at Risk), IDR (Incremental Default Risk) with VaR and IRC with ES (Expected Shortfall). We distinguish portfolio types (mixed, investment or speculative); the impact of different types of generators (stable vs unstable); dependence on the sample size and general convergence. We compare probabilities of default as maps of time across different algorithms and find interesting results. 
	
\end{enumerate}
For the study carried out, the implemented EM algorithm is very competitive. It is slightly slower than the deterministic algorithms but much faster than the MCMC algorithms. It embeds statistical properties like robustness that deterministic algorithms cannot capture.   

\begin{remark}
We focus purely on continuous over discrete time models. Continuous time algorithms yield robust estimators while the discrete ones do not, with robustness understood in the following sense: from $\mathbf{P}(1)$ estimate $\mathbf{P}(0.5)$ and $\mathbf{P}(0.25)$. From $\mathbf{P}(0.5)$ estimate $\mathbf{P}(0.25)$ again. Continuous algorithms yield the same $\mathbf{P}(0.25)$, discrete algorithms (in general) will not.
\end{remark}

\begin{remark}[Software availability]
\label{rem:pfeuffer-ctmcd}
The findings and algorithms of this work are now part of an improved version of the CRAN R-package \emph{ctmcd: Estimating the Parameters of a Continuous-Time Markov Chain from Discrete-Time Data} (see \cite{Pfeuffer2017}) --- {https://CRAN.R-project.org/package=ctmcd}
\end{remark}

This work is organized as follows. In Section \ref{sec:EMalgorithm} we present the EM algorithm and we state our main theoretical findings. In Section \ref{sec:CompetitorAlgorithms} we briefly present other known algorithms and in Section \ref{Sec:Comparing Algorithms} we present the benchmarking results.


\section{The EM Algorithm}
\label{sec:EMalgorithm}
There exists extensive literature on the majority of the algorithms we present in this paper, therefore we only provide brief discussions and include references for additional information. Further, we will use the theory of Markov chains extensively. We do not provide details of the theory, however, interested readers can consult texts such as \cite{Norris1998}.
\medskip

\noindent\textbf{Preliminaries and standing convention.}
Throughout this manuscript we consider companies defined on a finite state space $\{1, \dots, h\}$, where each state corresponds to a rating. We denote $AAA$ as rating $1$ and $D$ (default) as rating $h$. We adopt the standard notation that $\mathbf{P}$ is an $h$-\emph{by}-$h$ stochastic matrix, which will be the observed TPM (at, say, time $t=1$) and $\mathbf{Q}$ is an $h$-\emph{by}-$h$ generator matrix. We further denote by $P_{ij}:=\left(\mathbf{P}\right)_{ij}$, by $q_{ij}:=\left(\mathbf{Q}\right)_{ij}$ and the intensity of state $i$ by $q_{i}=\sum_{j \neq i}q_{ij}$ where $i,j \in \{1,\dots,h\}$. A standard assumption used in credit risk modelling is that default is an absorbing state, hence $P_{hh}=1$.
We work with infinitesimal generators of the following class.
\begin{definition}[Stable-Conservative infinitesimal Generator matrix of a CTMC]
\label{def:GeneratorMatrixSpace}
We say a matrix $\mathbf{Q}$ is a generator matrix if the following properties are satisfied for all $i,j \in \{1, \dots, h\}$: i)  $0 \le q_{ij} < \infty$ for $i \neq j$; ii) $q_{ii} \le 0$; and iii) $\sum_{j=1}^{h} q_{ij}=0$ $\forall i$.
\end{definition}
If a matrix $\mathbf{Q}$ satisfies the above properties, then for all $t \ge 0$ the matrix $\mathbf{P}(t):=e^{\mathbf{Q}t}$ is a stochastic matrix, where $e^{\mathbf{A}}$ is the matrix exponential of matrix $\mathbf{A}$ \cite{Norris1998}*{p.63}. The goal of the algorithms presented is to calculate a generator matrix $\mathbf{Q}$ such that $e^{\mathbf{Q}t}$ is the best fit to the observed TPM, where $t$ denotes the length of time between the rating updates (typically one year). 

Throughout let $(X(t))_{t\geq0}$ denote a CTMC over the finite state space $\{1, \dots, h\}$ with a generator $\mathbf{Q}$ of the above class. Associated to $X(t)$ is, for $i,j$ in the state space, $K_{ij}(t)$ the number of jumps from $i$ to $j$ in the interval $[0,t]$ and by $S_{i}(t)$ the holding time of state $i$ in the interval $[0,t]$.
\begin{remark}
\label{rem:embeddableP}
If a matrix $\mathbf{P}$ is embeddable\footnote{In this setting a stochastic matrix $\mathbf{P}$ is embeddable if there exists a generator $\mathbf{Q}$ such that $\mathbf{P}=e^{\mathbf{Q}}$.}, the algorithms below are pointless and one can easily tackle the problem through eigenvalue decomposition etc. Or in the case where the exact timing of rating transitions are known one can use the standard maximum likelihood estimator as in \cite{JarrowLandoTurnbull1997}.
In our examples the only data given is a set of yearly TPMs which in general are not embeddable and the methods just mentioned do not yield useful results.
\end{remark}

\subsection{The Algorithm}
Many methods have been developed in statistics in order to calculate the maximum likelihood estimate, but many methods break in the presence of missing data. Mathematically, we are interested in some set $\mathcal{X}$, but we are only able to observe $\mathcal{Y}$, with the assumption there is a many-to-one mapping from $\mathcal{X}$ to $\mathcal{Y}$. That is, $\cX$ is a much richer set than $\cY$. When dealing with such a case, the Expectation Maximization (EM) algorithm often offers a robust solution to the problem. \cite{McLachlanKrishnan2007} provide a complete overview of the algorithm.

The basis of algorithm is, we observe data $y$ which is a realisation (element) of $\mathcal{Y}$. We know $y$ has density function $g$ (sometimes referred to as a sampling density) depending on parameters $\Psi$ in some space $\Lambda$, but we want the density (likelihood) of $\mathcal{X}(y)$. Hence, postulate some family of densities $f$, dependent on $\Psi$, where $f$ corresponds to the density of the complete data set $\mathcal{X}(y)$ (the set of points $x \in \cX$ which are in the pre-image of $y\in \cY$). The relation between $f$ and $g$ is,
\begin{align*}
g(y; \Psi)
=
\int_{\mathcal{X}(y)}f(x;\Psi)dx \, . 
\end{align*}
The idea is, the EM algorithm maximizes $g$ w.r.t. $\Psi$, but we force it to do so by using the density $f$. Further, define,
\begin{align}
\label{Qdefn}
R(\Psi^{\prime};\Psi) 
:=
\mathbb{E}_{\Psi}\Big[\ln\big( f(x;\Psi^{\prime}) \big) \big| y\Big] \, 
\qquad \textrm{for }~ \Psi',\Psi \in \Lambda,
\end{align}
where $\bE_{\Psi}[\cdot|y]$ is the conditional expectation, conditional on $y$ under parameters $\Psi$.
We assume $R$ to exist for all pairs $(\Psi^{\prime},\Psi)$, in particular we assume $f(x;\Psi) >0$ almost everywhere in $\mathcal{X}$ for all $\Psi$ (otherwise the logarithm is infinite). Let us clarify, $f$ is calculated using $\Psi^{\prime}$, but the expectation is calculated using $\Psi$. The EM algorithm is then the following iterative procedure. 
\begin{enumerate}
	\item Choose an initial $\Psi^{(1)}$ and take $p=1$.
	\item E-step: Compute $R(\Psi;\Psi^{(p)})$.
	\item M-step: Choose $\Psi^{(p+1)}$ to be the value of $\Psi \in \Lambda$ that maximizes $R(\Psi;\Psi^{(p)})$.
	\item Check if the predefined convergence criteria is met, if not, take $p=p+1$ and return to (ii).
\end{enumerate}

\subsubsection{The particular problem of generator estimation}

For our problem the observed process is a discrete time Markov chain (DTMC), the unobserved process to estimate is a continuous time Markov chain (CTMC). Therefore, the observed data is the discrete transitions and the parameters we wish to estimate are the entries in the generator. The likelihood of a continuous time fully observed Markov chain with generator $\mathbf{Q}$ is given by the following expression (see \cite{KuchlerSorensen1997}*{Chapter 3.4}),
\begin{align*}
L_{t}(\mathbf{Q})=\exp \left\{\sum_{i=1}^{h} \left[\sum_{j \neq i}\log(q_{ij}) K_{ij}(t) -  S_{i}(t)\sum_{j \neq i}q_{ij} \right]\right \} \, ,
\end{align*}
with $K$ and $S$ the same as before (immediately before Remark \ref{rem:embeddableP}). Hence given two generators $\mathbf{Q}',\mathbf{Q}$, the function $R$ in \eqref{Qdefn} is,
\begin{align}
\label{Eq:R function for CTMC}
R(\mathbf{Q}^{\prime};\mathbf{Q})
=
\sum_{i=1}^{h} \left[
                     \sum_{j \neq i}\log(q'_{ij}) \bE_{\mathbf{Q}}[K_{ij}(t)|y] 
										  -  \bE_{\mathbf{Q}}[S_{i}(t)|y]\sum_{j \neq i}q'_{ij} 
              \right] \, , 
\end{align}
where $y$ denotes the discrete time observations. Maximizing for $q'_{ij}$ in $R(\mathbf{Q}';\mathbf{Q})$ yields 
\begin{align}
\label{Eq:Max EM Function}
q'_{ij}=\frac{\bE_{\mathbf{Q}}[K_{ij}(t)|y]}{ \bE_{\mathbf{Q}}[S_{i}(t)|y]} \, .
\end{align}
The difficult step is the calculation of $\bE_{\mathbf{Q}}[K_{ij}(t)|y]$ and $ \bE_{\mathbf{Q}}[S_{i}(t)|y]$. We follow an approach similar to \cite{BladtSorensen2005} (see also \cite{BladtEtAl2002}) but express the result in a framework more suited to the problem of generator estimation from TPMs, rather than the estimation from individual movements. Furthermore, the result derived in \cite{BladtEtAl2002} is for irreducible Markov chains making it not applicable to our case (CTMC with absorbing states),  accounted for in Proposition \ref{Prop:Expected Values}. 

Consider the following functions (see \cite{BladtEtAl2002}), for $1 \le i,j \le h$
\begin{align}
\label{LaplaceProbTransform}
V^{*}_{ij}(\mathbf{c},\mathbf{Z};t)
=
\bE_{\mathbf{Q}}\left[\exp \left\{-\sum_{\mu=1}^{h}c_{\mu}S_{\mu}(t) \right\}
\prod_{\mu, \nu=1}^{h}Z_{\mu \nu}^{K_{\mu \nu}(t)}\mathbbm{1}_{\{X(t)=j\}}
\biggl| X(0)=i \right] \, , 
\end{align}
where we denote by $\mathbf{c}=(c_1,\cdots,c_h)\in \bR^h$ and $\mathbf{Z}\in\bR^{h\times h}$ with $Z_{ii}=1$ for $i\in\{1,\cdots,h\}$.  Observe that $V^{*}_{ij}$ is the Laplace-Stieltjes transform of the holding times $S$ and the probability generating function of the jumps $K$, with initial and final states $X(0)=i$ and $X(t)=j$ respectively. Denoting by $\mathbf{V}^{*}(\mathbf{c},\mathbf{Z};t)$ the $h$-\emph{by}-$h$ matrix of elements $V^{*}_{ij}(\mathbf{c},\mathbf{Z};t)$. This allows us to give the main theorem (similar version) in \cite{BladtEtAl2002}.

\begin{theorem}
\label{Thm: MatrixExpRelation}
	For $t \ge0$, the matrix $\mathbf{V}^{*}(\mathbf{c},\mathbf{Z};t)$ is given by,
	\begin{equation*}
	\mathbf{V}^{*}(\mathbf{c},\mathbf{Z};t) = \exp\{[\mathbf{Q} \bullet \mathbf{Z} - \Delta(\mathbf{c})]t\} \, ,
	\end{equation*}
	where $\bullet$ is the Schur (Hadamard) product\footnote{The Shur product of two $h\times h$ matrices $A$ and $B$ is the $h\times h$ matrix with elements $A_{ij}B_{ij}$.} of matrices, $\mathbf{Q}$ is the generator matrix from the CTMC and $\Delta(\mathbf{c})$ is the diagonal matrix with entries $c_i$ at position $ii$ for $i=1,\dots,h$.
\end{theorem}
A closed form expression for the expectation terms in \eqref{Eq:Max EM Function} follows from a result in \cite{VanLoan1978} (sketched also in \cite{HobolthJensen2011}).
\begin{proposition} 
\label{Prop:Expected Values}
	Let $\textbf{e}_{i}$ be the column vector of length $h$ which is one at entry $i$ and zero elsewhere, further let us define the $2h$-by-$2h$ matrices $\mathbf{C}^{(\alpha \beta)}_{\gamma}$ and $\mathbf{C}^{(\alpha)}_{\phi}$ as,
	\begin{equation}
	\label{Eq:AlphaandBetaMatrices}
	\mathbf{C}^{(\alpha \beta)}_{\gamma}=
	\left[ {\begin{array}{cc}
			\mathbf{Q} & q_{\alpha \beta}\mathbf{e}_{\alpha} \mathbf{e}_{\beta}^{\intercal} \\
			0 & \mathbf{Q}
		\end{array} } \right]
	\qquad
	 \text{ and }
	 \qquad  
	 \mathbf{C}^{(\alpha)}_{\phi}=
	 \left[ {\begin{array}{cc}
	 		\mathbf{Q} & \mathbf{e}_{\alpha} \mathbf{e}_{\alpha}^{\intercal} \\
	 		0 & \mathbf{Q}
	 	\end{array} } \right]\,
		\quad \alpha,\beta\in\{1,\cdots,h\}. 
	\end{equation}
Consider a CTMC $X$ that we observe at $n$ time points $0 \le t_1 < t_2 < \dots < t_n$ and denote by $y_s$ the state of the Markov chain at time $t_s$, i.e. $y_s := X(t_s)$.
	Then, the expected jumps and holding times over the observations are,
	\begin{align*}
	\bE_{\mathbf{Q}}[K_{ij}(t)|y]
	=
	& 
	\sum_{s=1}^{n-1} \frac{\left(\exp\{\mathbf{C}^{(i j)}_{\gamma} (t_{s+1} -t_{s})\}\right)_{y_{s},h+y_{s+1}}}{\left(\exp\{\mathbf{Q} (t_{s+1} -t_{s})\}\right)_{y_{s},y_{s+1}}} \, , 
	\\
	\bE_{\mathbf{Q}}[S_{i}(t)|y]=
	& 
	\sum_{s=1}^{n-1} \frac{\left(\exp\{\mathbf{C}^{(i)}_{\phi} (t_{s+1} -t_{s})\}\right)_{y_{s},h+y_{s+1}}}{\left(\exp\{\mathbf{Q} (t_{s+1} -t_{s})\}\right)_{y_{s},y_{s+1}}}.
	\end{align*}
\end{proposition}

\begin{proof}
Observe that $\mathbf{V}^{*}$ in \eqref{LaplaceProbTransform} satisfies the differential equation in the statement of Theorem \ref{Thm: MatrixExpRelation} (see \cite{BladtSorensen2005}). The proof is omitted as one can solve the equation by employing the methods in \cite{VanLoan1978}.
\end{proof}
Thus we obtain closed form expressions for the two key quantities appearing in \eqref{Eq:Max EM Function}. This approach differs from \cite{BladtSorensen2005} where they describe numerical schemes to solve the differential equations, namely Runge-Kutta and uniformization. These techniques can yield good results at this level, but our closed form expression will pay dividends when it comes to error estimation.

This yields the relation we desire, however, in our example we have an observed TPM (or sequence of TPMs), $\mathbf{P}$, in the case of equal observation windows, $t$ in the interval $[0,T]$ (although it is trivial to generalize) the expectation can be expressed as,
\begin{align}
\label{Eq:ExpectedHoldandJumpsTPM}
\bE_{\mathbf{Q}}[K_{ij}(T)|\mathbf{P}]
=
& \sum_{u=1}^{N} \sum_{s=1}^{h}\sum_{r =1}^{h}P^{u}_{sr}(t)\frac{\left(\exp\{\mathbf{C}^{(i j)}_{\gamma} t\}\right)_{s,h+r}}{\left(\exp\{\mathbf{Q} t\}\right)_{s,r}} \, , \notag
\\
\bE_{\mathbf{Q}}[S_{i}(T)|\mathbf{P}]
=
& 
\sum_{u=1}^{N} \sum_{s=1}^{h}\sum_{r=1}^{h} P^{u}_{sr}(t) \frac{\left(\exp\{\mathbf{C}^{(i)}_{\phi}t\}\right)_{s,h+r}}{\left(\exp\{\mathbf{Q} t\}\right)_{s,r}} \, , 
\end{align}
where $N=T/t$ (the number of observations) and $\mathbf{P}^{u}$ is the TPM of the $u$-th observation.

\begin{remark}[The reducible case]
Previously, we only had observed transitions, hence they must have a non-zero probability of occurring. Here we can sum $s$ and $r$ over the full range because $\mathbf{P}_{sr}(t)$ acts as an indicator of possible transitions, that is, if $P_{sr}(t)=0$ we set the $s, \, r$ component as $0$. Clearly, if $P_{sr}(t)>0$, but $(\exp\{\mathbf{Q}t\})_{s r}=0$, $\mathbf{Q}$ is misspecified.  
\end{remark}

\subsubsection*{Likelihood Convergence of the EM algorithm}

In the case of this problem \cite{BladtSorensen2005} provide a proof that the likelihood function converges with one small caveat in order to keep the parameter space compact. Namely, they use the following constrained parameter space, $\cQ_{\epsilon}$, which can be achieved by setting, $\cQ_{\epsilon}=\{\mathbf{Q}\in \cQ|\text{det}[\exp(\mathbf{Q})] \ge \epsilon\}$ ($\cQ$ is the parameter space from Definition \ref{def:GeneratorMatrixSpace}) for some $\epsilon >0$. Theorem 4 in \cite{BladtSorensen2005} states that the algorithm will converge to a stationary point of the likelihood or hit the boundary of the parameter space they have induced. It is accepted this is a crude approach to solving the problem and further analysis is needed when $\text{det}[\exp(\mathbf{Q})] = \epsilon$. An alternative approach would be to use a penalized likelihood as discussed in \cite{McLachlanKrishnan2007}*{p.214}.

\subsubsection*{Parameter convergence criteria}

The above convergence is sufficient for one to conclude convergence of the likelihood. However, it is not sufficient for convergence of the parameters as one cannot state that the series of iterates $\mathbf{Q}^{(k)}$ converge ($\|\mathbf{Q}^{(k+1)}-\mathbf{Q}^{(k)}\| \to 0$ as $k\to\infty$). From a theoretical standpoint this may not be as important as convergence of the likelihood itself, nonetheless, it is of key importance for applications. For instance, without convergence of the parameters the risk charge different financial agents obtain from the same data may vary wildly, even under very strict convergence conditions. Before proving convergence we require two important points.

\begin{remark}
	With \eqref{Eq:ExpectedHoldandJumpsTPM} in mind we assume that for any $s \neq r$ such that $P^{u}_{sr}(t)=0$ for all $u$, we take the starting point $q_{sr}^{(0)}:=(\mathbf{Q}^{(0)})_{sr}=0$. As discussed in \cite{BladtSorensen2005}, any point set to zero will stay at zero for all iterations. Note, we are not changing the problem since these terms will converge to zero under the EM algorithm.
\end{remark}

\begin{assumption}[Element constraint]
	\label{Assump:Element Constraint}
	Similar to \cite{BladtSorensen2005}, we will use a manual space constraint to obtain the convergence. Take $1>\epsilon > 0$, such that $\forall ~ i\neq j$, $q_{ij}< 1/\epsilon$. Moreover, we assume adjacent mixing, namely, for $i \in \{2,\dots, h-1\}$, $q_{i, i\pm 1}>\epsilon$ and $q_{1,2}> \epsilon$.  
	
	We denote the space of generator matrices which satisfy this condition as $\Lambda_{\epsilon}$.
\end{assumption}
The above assumption ensures non-zero entries in the tri-diagonal band and also only finite entries as one can take $\epsilon$ as small as we wish. In the case of TPMs associated to credit ratings, such an assumption is trivially satisfied as one generally has diagonally dominant matrices and companies can always be upgraded or downgraded by one, thus $P^{u}_{i,i\pm 1}$ are typically non-zero. Diagonal dominance is sufficient for the generator to be identifiable and therefore entries do not blow up, we discuss the notion of identifiability in Section \ref{Sec:Estimation of Error}.

Proving the parameters converge is more challenging than the likelihoods, however, \cite{Wu1983} provide a sufficient condition for this to occur, namely a sufficient condition for $\|\mathbf{Q}^{(k+1)}-\mathbf{Q}^{(k)}\| \to 0$ as $k\to\infty$ is, there exists a forcing function\footnote{A forcing function is defined as any function $\sigma: [0,\infty) \rightarrow [0,\infty)$ such that for any sequence ${t_{k}}$ defined in $[0,\infty)$, $\lim_{k \rightarrow \infty} \sigma(t_{k})=0$ implies $\lim_{k \rightarrow \infty} t_{k}=0$.} $\sigma$ such that,
\begin{align*}
R(\mathbf{Q}^{(k+1)};\mathbf{Q}^{(k)})-R(\mathbf{Q}^{(k)};\mathbf{Q}^{(k)}) \ge \sigma(||\mathbf{Q}^{(k+1)}-\mathbf{Q}^{(k)}||) \, .
\end{align*}
An example of a forcing function is $\sigma(t)=\lambda t^2$ where $\lambda >0$. We require the following bounds on the expected values to show convergence.

\begin{lemma} 
\label{Lem:ExpectationBounds}
Let $N$ and $\mathbf{P}^{u}$ be as defined in \eqref{Eq:ExpectedHoldandJumpsTPM} and assume for $i \neq j$ there exists a $u \in \{1, \dots, N\}$ such that $P^{u}_{ij}>0$ (we observe a movement from $i$ to $j$ in observation window $u$). Then we obtain the following bounds on the expected number of jumps:
\begin{align}
\label{Eq:Bounds on expected jumps}
P^{u}_{ij} \frac{\epsilon q_{ij}}{h} 
\le 
\bE_{\mathbf{Q}}[K_{ij}(T)|\mathbf{P}] 
\le
 h^2 N \frac{ht}{\epsilon \min\big\{\epsilon^{h}t^{h}\exp\{-th^{2}/\epsilon\}\ ,\ \exp\{ht/ \epsilon \} \big\}}\, .
\end{align}
Moreover, assuming there exists a $u \in \{1, \dots, N\}$ such that $P^{u}_{ii}>0$, we obtain the following bound on the expected holding time,
\begin{align}
\label{Eq:Bound on expected holding}
 \bE_{\mathbf{Q}}[S_{i}(T)|\mathbf{P}] \ge P^{u}_{ii} t \exp\big\{ -\frac{ht}\epsilon \big\} \, .
\end{align}
\end{lemma}
To maintain the flow of the text we state immediately our main convergence result, and defer the proof of the Lemma to Appendix \ref{Sec:Proof of Expectation Lemma}. 
\begin{theorem}
\label{Thm:Single Point Convergence}
Under Assumption \ref{Assump:Element Constraint}, then, there exists a $\lambda>0$ such that for all EM iterations $k\in \bN$,
\begin{equation*}
R\big(\mathbf{Q}^{(k+1)};\mathbf{Q}^{(k)}\big)-R\big(\mathbf{Q}^{(k)};\mathbf{Q}^{(k)}\big) \ge \lambda \|\mathbf{Q}^{(k+1)}-\mathbf{Q}^{(k)}\|^2 \, ,
\end{equation*}
where $||\cdot||$ is the Euclidean norm.
\end{theorem}

\begin{proof}
	Writing out the difference in the $R$ terms we obtain,
	\begin{align*}
	\sum_{i=1}^{h} \sum_{j \neq i}\left[
	 \bE_{\mathbf{Q}^{(k)}}[K_{ij}(t)|\mathbf{P}] \big(\log(q^{(k+1)}_{ij})-\log(q^{(k)}_{ij})\big)-\bE_{\mathbf{Q}^{(k)}}[S_{i}(T)|\mathbf{P}]\big(q^{(k+1)}_{ij}- q^{(k)}_{ij}\big)
	\right] \, .
	\end{align*}	
Due to the form of the Euclidean norm squared and the function $R$, it is sufficient to show the inequality holds for all $i \neq j$. Namely, it is sufficient to show the existence of a $\lambda>0$ such that,
\begin{align}
\label{Eq:R Difference}
\nonumber
\bE_{\mathbf{Q}^{(k)}}[K_{ij}(T)|\mathbf{P}]\Big( & \log(q^{(k+1)}_{ij}) -\log(q^{(k)}_{ij})\Big)
\\
&
-\bE_{\mathbf{Q}^{(k)}}[S_{i}(T)|\mathbf{P}]\left(q^{(k+1)}_{ij}- q^{(k)}_{ij} \right) 
\ge
 \lambda \left(q^{(k+1)}_{ij}-q^{(k)}_{ij} \right)^2 \, ,
\end{align}
for all $i \neq j$. We tackle the log terms first. It is well known that we can express any $C^{\infty}$-function using Taylor expansion to a finite number of terms with some error (remainder) term. Moreover, the error term has a known form and hence, using an exact Taylor expansion to second order, there exists a $Z \in [\min(q^{(k)}_{ij},q^{(k+1)}_{ij}), \max(q^{(k)}_{ij},q^{(k+1)}_{ij})]$ such that,
\begin{align*}
\log(q^{(k+1)}_{ij})-\log(q^{(k)}_{ij})=\frac{-1}{q^{(k+1)}_{ij}}(q^{(k)}_{ij}-q^{(k+1)}_{ij}) +\frac{1}{2Z^2}(q^{(k)}_{ij}-q^{(k+1)}_{ij})^2 \, ,
\end{align*}
where we have expanded $q^{(k)}_{ij}$ around $q^{(k+1)}_{ij}$. Substituting \eqref{Eq:Max EM Function} into the LHS of \eqref{Eq:R Difference}, the condition simplifies to,
\begin{align*}
\frac{\bE_{\mathbf{Q}^{(k)}}[K_{ij}(T)|\mathbf{P}]}{2Z^2}(q^{(k)}_{ij}-q^{(k+1)}_{ij})^2 \ge \lambda (q^{(k+1)}_{ij}-q^{(k)}_{ij})^2 \, .
\end{align*}
In order to show this bound we need to get a handle on $Z$. Clearly, there are two options between iterations, either $q^{(k)}_{ij}>q^{(k+1)}_{ij}$ or $q^{(k)}_{ij} \le q^{(k+1)}_{ij}$. In the latter case we obtain,
\begin{align*}
\frac{\bE_{\mathbf{Q}^{(k)}}[K_{ij}(T)|\mathbf{P}]}{2Z^2}(q^{(k)}_{ij}-q^{(k+1)}_{ij})^2 \ge \frac{\bE_{\mathbf{Q}^{(k)}}[S_{i}(T)|\mathbf{P}]^2}{2\bE_{\mathbf{Q}^{(k)}}[K_{ij}(T)|\mathbf{P}]}(q^{(k)}_{ij}-q^{(k+1)}_{ij})^2 \, .
\end{align*}
Since we can element wise bound $\mathbf{Q}^{(k)}$, using Lemma \ref{Lem:ExpectationBounds} and Assumption \ref{Assump:Element Constraint} we can bound the term $\bE_{\mathbf{Q}^{(k)}}[K_{ij}(T)|\mathbf{P}]$ from above and $\bE_{\mathbf{Q}^{(k)}}[S_{i}(T)|\mathbf{P}]$ from below by constants (depending on $\epsilon$). Hence, we can choose a $\lambda$ independent of $k$ such that the condition is satisfied.

The second case  $q^{(k)}_{ij}>q^{(k+1)}_{ij}$, follows a similar argument. Again, we can set $Z$ as the larger of the two values, thus we obtain the following inequality,
\begin{align*}
\frac{\bE_{\mathbf{Q}^{(k)}}[K_{ij}(T)|\mathbf{P}]}{2Z^2}(q^{(k)}_{ij}-q^{(k+1)}_{ij})^2 \ge \frac{\bE_{\mathbf{Q}^{(k)}}[K_{ij}(T)|\mathbf{P}]}{2 \left(q_{ij}^{(k)}\right)^2}(q^{(k)}_{ij}-q^{(k+1)}_{ij})^2 \, .
\end{align*}
Using Lemma \ref{Lem:ExpectationBounds}, we can reduce this inequality to,
\begin{align*}
\frac{\bE_{\mathbf{Q}^{(k)}}[K_{ij}(T)|\mathbf{P}]}{2Z^2}(q^{(k)}_{ij}-q^{(k+1)}_{ij})^2 \ge \frac{P^{u}_{ij} \epsilon}{2 h q_{ij}^{(k)}}(q^{(k)}_{ij}-q^{(k+1)}_{ij})^2 \, .
\end{align*}
Since $P^{u}_{ij}>0$ and we can bound each $q_{ij}$ from above, again we choose a $\lambda$ independent of $k$.
\end{proof}

\subsubsection*{Starting value for the EM algorithm}
The final point to discuss, is the choice of the initial matrix $\mathbf{Q}$. It is useful from a computational point of view to start in a good place. Here we choose $\mathbf{Q}$ based on a generalization of the QOG algorithm (described in Section \ref{sec:DeterministicAlgorithms}) that allows for complex inputs. For each entry $q_{ij}$ we define the input as,
\begin{equation*}
q_{ij} \rightarrow \text{sign}\big(\text{Re} (q_{ij})\big)\times|q_{ij}| \, ,
\end{equation*}
where $|q_{ij}|$ is the magnitude of $q_{ij}$ and Re$(q_{ij})$, is the real component of $q_{ij}$. With the newly defined $\mathbf{Q}$ we apply the QOG algorithm. We take any zero entries not in the final row to be a small number ($10^{-5}$, say) unless there are zero observed transitions. This defines our initial choice of $\mathbf{Q}$. We define the EM algorithm steps as,
\begin{enumerate}
\item Take an initial intensity matrix $\mathbf{Q}$ and positive value $\epsilon$.
\item While the convergence criteria is not met and all entries of $\mathbf{Q}$ are within the boundaries
\begin{enumerate}[(1)]
\item E-step: calculate $\bE_{\mathbf{Q}}[K_{ij}(T)|\mathbf{P}]$ and $\bE_{\mathbf{Q}}[S_{i}(T)|\mathbf{P}]$.
\item M-step: set $q'_{ij}=\bE_{\mathbf{Q}}[K_{ij}(T)|\mathbf{P}]/\bE_{\mathbf{Q}}[S_{i}(T)|\mathbf{P}]$, for all $i \neq j$ and set $q_{ii}$ appropriately.
\item Set $\mathbf{Q}=\mathbf{Q}'$ (where $\mathbf{Q}'$ is the matrix of $q'$s) and return to E-step.
\end{enumerate}
\item End while and return $\mathbf{Q}$.
\end{enumerate}
This leads to the following theorem for convergence of the EM. 
\begin{theorem}[Convergence of the EM]
	\label{Thm:Convergence of the EM}
	Assume that our initial point is in the parameter space $\Lambda_\epsilon$: is a true generator and satisfies Assumption \ref{Assump:Element Constraint}. Then either the sequence of points $\{\mathbf{Q}^{(k)}\}_{k}$ converges to a single point in $\Lambda_\epsilon$ which is also a stationary point of the likelihood, or the entries go to the boundary (blow up or some tri-diagonal elements in an non-absorbing row go to zero).
\end{theorem}
A proof of Theorem \ref{Thm:Convergence of the EM} follows directly from Theorem 4 in \cite{BladtSorensen2005} and our Theorem \ref{Thm:Single Point Convergence}.

\begin{remark}[The unique maximizer of the Likelihood]
\label{rem:LikelihoodUniqueMaximizer}
The natural question one may ask is does this stronger form of convergence imply convergence to the global maximum? The problem of existence and uniqueness of maximum likelihoods in this setting is a very challenging problem with a long history. \cite{BladtSorensen2005} give a wonderful overview on the subject, Theorem 1 in \cite{BladtSorensen2005} also provides results on existence and uniqueness of the maximum. Unfortunately, one cannot say more than this, if one can derive conditions under which a unique maximum existed (for non-embeddable TPMs) then the above convergence result is sufficient to conclude the EM will converge to the MLE.

During the final revision of this manuscript, \cite{PfeufferMoestelFischer2017} came to our attention. There, the authors propose two algorithms to mitigate the effect of the EM's possible convergence towards a local but not necessarily the global maximum of the likelihood function (no proofs are given). Our Theorem \ref{Thm:Single Point Convergence} is handy in this context as it shows that once the EM lands ``near'' the global maximum the iteration will converge to it. 
\end{remark}

\subsection{Variance Estimation} 
\label{Sec:Estimation of Error}

In this section we derive an expression for the Hessian of the likelihood. We use a result in \cite{Oakes1999} and follow \cite{BladtSorensen2009}, however, unlike \cite{BladtSorensen2009}, we provide a closed form expression for the Hessian. This result eliminates the stability problems observed in the numerical simulation case when the entries in $\mathbf{Q}$ are small. 
The Hessian provides a way to estimate the standard error of the maximum likelihood estimates and further allows us to assess the nature of the converged stationary point (this is further discussed in Section \ref{Sec:Connection to Global Max}).


We point the reader to \cite{BladtSorensen2005}*{Theorem 1} for results on the existence and uniqueness of maximum likelihood estimators with respect to this problem. Further, for discussions on consistency and asymptotic normality related to this problem one should consult \cite{KremerWeissbach2013}, \cite{KremerWeissbach2014}. \cite{KremerWeissbach2013}, provide sufficient conditions for consistency, the key assumption relies on so-called model \emph{identifiability}\footnote{In our setting a model is identifiable if there are no two generators $\mathbf{Q} \neq \mathbf{Q}'$ such that $\exp\{\mathbf{Q}t\}=\exp\{\mathbf{Q}'t\}$.}. \cite{KremerWeissbach2013} prove \emph{identifiability} under conditions which are too restrictive for our purpose; \cites{BladtSorensen2005, DehayYao2007} discusses the problem of \emph{identifiability} in detail. From \cite{Cuthbert1973}, \cite{BladtSorensen2005} for the model to be identifiable it is sufficient (though very crude) to have $\min_{i}(\exp\{\mathbf{Q}t\})_{ii}> 1/2$, \cite{Culver1966} gives a requirement for general matrices based on the eigenvalues which one can always aposteriori verify after a $Q$ is deduced. The crucial assumption in \cite{KremerWeissbach2014} to obtain asymptotic normality, is that the Hessian must be invertible at the true value, we can of course verify invertibility a posteriori.


We recall a result from \cite{Oakes1999} for calculating the Hessian of the likelihood.
\begin{lemma} \label{Lem: EM Hessian}
The second derivative of the likelihood with parameter $\Psi$ and observed information $y$ is related to the EM function $R$ by
\begin{equation*}
\frac{\partial^{2}L(\Psi;y)}{\partial \Psi^{2}}= \left[ \frac{\partial^{2}R(\Psi' ;\Psi)}{\partial \Psi'^{2}} +\frac{\partial^{2}R(\Psi';\Psi)}{\partial \Psi' \partial\Psi}\right]_{\Psi'=\Psi} \, .
\end{equation*}
\end{lemma}
Injecting \eqref{Eq:R function for CTMC} in the above we obtain,
\begin{equation}
\label{Eq:DoublePrimeDerivative}
\dfrac{\partial^{2}R(\mathbf{Q}^{\prime};\mathbf{Q})}{\partial q'_{\alpha \beta} \partial q'_{\mu \nu}} = \frac{-1}{q'^{2}_{\mu \nu}} \bE_{\mathbf{Q}}[K_{\mu \nu}(t)|y] \delta_{\alpha \mu}\delta_{\beta \nu} \, , 
\end{equation}
\begin{equation}
\label{Eq:SinglePrimeDerivative}
\dfrac{\partial^{2}R(\mathbf{Q}^{\prime};\mathbf{Q})}{\partial q_{\alpha \beta} \partial q'_{\mu \nu}} = \frac{1}{q'_{\mu \nu}} \dfrac{\partial}{\partial q_{\alpha \beta}}\bE_{\mathbf{Q}}[K_{\mu \nu}(t)|y] - \dfrac{\partial}{\partial q_{\alpha \beta}} \bE_{\mathbf{Q}}[S_{\mu}(t)|y] \, , 
\end{equation}
where $\delta_{ab}$ is the Kronecker delta. From our previous work, \eqref{Eq:DoublePrimeDerivative} is easy to obtain, however, \eqref{Eq:SinglePrimeDerivative} involves derivatives of the expected jumps and holding times and is thus challenging. \cite{BladtSorensen2009} opt for a simple numerical scheme to compute these derivatives and found unstable results, although the authors do remark that more sophisticated numerical schemes could yield improved results at greater computational expense.

It is worth noting we have made no comment on the allowed values of $\alpha, \beta, \mu$ and $\nu$, other than the clear restriction that they belong to $ \{1,\dots,h\}$. Let us now state the following definition. 
\begin{definition}[Allowed pairs]
\label{def:AllowedPairs}
We say that the pair $\alpha, \beta$ is allowed if $\alpha\neq \beta$ (not in the  diagonal) and $q_{\alpha \beta}$ is not converging to zero under the EM algorithm. 
\end{definition}
For practical applications, one can imagine the set of allowed values, as the set of $\alpha, \beta$ such that $q_{\alpha \beta}>\epsilon$, where $\epsilon$ is some cut-off point ($10^{-8}$, say). The reason we must exclude small parameters is, this analysis only holds in the large data limit, since we do not have an infinite amount of data we cannot for certain rule out some jump, however, if $q_{\alpha \beta}$ is converging to zero, it implies that this parameter is either zero or extremely close to zero and therefore we can bound it above by a small number. Moreover, from a mathematical point of view a parameter which does tend to zero (or even becomes zero) lies on the boundary, where the notion of differentiability is not clear. Therefore, we can think of the ``allowed pairs'' as the variables when solving the problem in a restricted parameter space. We now present the following theorem.

\begin{theorem}
\label{theo:SecondDerivativesR}
	Let $\mu,\nu,\alpha,\beta \in \{1,\dots,h\}$, and $\mathbf{Q}$, $\mathbf{Q}'$ be two generator matrices ($\in \Lambda_{\epsilon}$ for some $\epsilon>0$). For any two allowed pairs $\alpha, \beta$ and $\mu, \nu$, the derivative in \eqref{Eq:SinglePrimeDerivative} is, 
		\begin{align*}
	\dfrac{\partial^{2}R(\mathbf{Q}^{\prime};\mathbf{Q})}{\partial q_{\alpha \beta} \partial q'_{\mu \nu}} 
	=
	&
	\sum_{s=1}^{n-1}\frac{1}{q'_{\mu \nu}}
	     \biggl[-(e^{\mathbf{Q}(t_{s+1}-t_{s})})^{-2}_{y_{s},y_{s+1}} 
			         \left(e^{\mathbf{C}^{(\alpha \beta)}_{\eta}(t_{s+1}-t_{s})} \right)_{y_{s}, h+y_{s+1}}
							  (e^{\mathbf{C}^{(\mu \nu)}_{\gamma}(t_{s+1}-t_{s})})_{y_{s}, h+y_{s+1}} 
	\notag  
	\\
	&
	\qquad\qquad \qquad 
	            + (e^{\mathbf{Q}(t_{s+1}-t_{s})})^{-1}_{y_{s},y_{s+1}}
							\left(e^{\mathbf{C}_{\psi}^{(\alpha \beta,\mu \nu)}(t_{s+1}-t_{s})} \right)_{y_{s}, 3h+y_{s+1}} 
			 \biggr] 
	\notag 
	\\
	&
	\qquad
	-\biggl[-(e^{\mathbf{Q}(t_{s+1}-t_{s})})^{-2}_{y_{s},y_{s+1}}
	         \left(e^{\mathbf{C}^{(\alpha \beta)}_{\eta}(t_{s+1}-t_{s})} \right)_{y_{s}, h+y_{s+1}}
					 (e^{\mathbf{C}^{(\mu )}_{\phi}(t_{s+1}-t_{s})})_{y_{s}, h+y_{s+1}} 
	\notag 
	\\
	&
	\qquad\qquad \qquad 
	          +\, (e^{\mathbf{Q}(t_{s+1}-t_{s})})^{-1}_{y_{s},y_{s+1}}
						\left(e^{\mathbf{C}_{\omega}^{(\alpha \beta,\mu)}(t_{s+1}-t_{s})} \right)_{y_{s}, 3h+y_{s+1}}
						\biggr] \, ,
	\end{align*}
	where the $2h$-by-$2h$ matrices, $\mathbf{C}^{(\alpha \beta)}_{\gamma}, \mathbf{C}^{(\alpha)}_{\phi}, \mathbf{C}^{(\alpha \beta)}_{\eta}$, are defined as,
	\begin{equation*}
	\mathbf{C}^{(\alpha \beta)}_{\gamma}=
	\left[ {\begin{array}{cc}
		\mathbf{Q} & q_{\alpha \beta}\mathbf{e}_{\alpha} \mathbf{e}_{\beta}^{\intercal} \\
		0 & \mathbf{Q}
		\end{array} } \right] ~ , ~ 
		\mathbf{C}^{(\alpha)}_{\phi}=
	\left[ {\begin{array}{cc}
		\mathbf{Q} & \mathbf{e}_{\alpha} \mathbf{e}_{\alpha}^{\intercal} \\
		0 & \mathbf{Q}
		\end{array} } \right] ~ , ~
	\mathbf{C}^{(\alpha \beta)}_{\eta}=
	\left[ {\begin{array}{cc}
		\mathbf{Q} & \mathbf{e}_{\alpha} \mathbf{e}_{\beta}^{\intercal} -\mathbf{e}_{\alpha}\mathbf{e}_{\alpha}^{\intercal}\\
		0 & \mathbf{Q}
		\end{array} } \right] \, ,
	\end{equation*}
	and the $4h$-by-$4h$ matrices $\mathbf{C}_{\psi}^{(\alpha \beta,\mu \nu)},\mathbf{C}_{\omega}^{(\alpha \beta,\mu)}$ are defined 
	\begin{equation*}
	 \mathbf{C}_{\psi}^{(\alpha \beta,\mu \nu)}
	=
	 \left[ {\begin{array}{cc}
	 	\mathbf{C}_{\gamma}^{(\mu \nu)} & \frac{\partial \mathbf{C}_{\gamma}^{(\mu \nu)}}{\partial q_{\alpha \beta}} \\
	 	0 & \mathbf{C}_{\gamma}^{(\mu \nu)}
	 	\end{array} } \right] ~ , ~ 
		\mathbf{C}_{\omega}^{(\alpha \beta,\mu)}
		=
	 \left[ {\begin{array}{cc}
	 	\mathbf{C}_{\phi}^{(\mu)} & \frac{\partial \mathbf{C}_{\phi}^{(\mu)}}{\partial q_{\alpha \beta}} \\
	 	0 & \mathbf{C}_{\phi}^{(\mu)}
	 	\end{array} } \right] \, .
	\end{equation*}
\end{theorem} 

The proof of this uses similar techniques to Proposition \ref{Prop:Expected Values} along with differentiation properties of matrix-exponentials, and is deferred to Appendix \ref{Sec:Proof of second derivatives}.

\begin{remark}
\label{remark:extraSizeOfMatrix}
In the above representation for the derivative of $R$, we use subscripts of the form $h+y_{s+1}$ and $3h+y_{s+1}$, this is simply a consequence of the result in \cite{VanLoan1978}. Namely, we are not interested in all the entries of the matrix, only an $h$-by-$h$ segment. We therefore need to adjust the indexing to only take elements at this specific segment.
\end{remark}

Using Theorem \ref{theo:SecondDerivativesR} and Lemma \ref{Lem: EM Hessian}, we can write the elements of the Hessian corresponding to the $q_{\alpha \beta} q_{\mu \nu}$ derivative as, 
\begin{align*}
\frac{\partial^{2}L(\mathbf{Q} ;y)}{\partial q_{\alpha \beta} \partial q_{\mu \nu}}=
&
\sum_{s=1}^{n-1}\frac{-1}{q^{2}_{\mu \nu}}(e^{\mathbf{Q}(t_{s+1}-t_{s})})^{-1}_{y_{s},y_{s+1}}(e^{\mathbf{C}^{(\mu \nu)}_{\gamma}(t_{s+1}-t_{s})})_{y_{s}, h+y_{s+1}} \delta_{\alpha \mu}\delta_{\beta \nu} 
\\
&
\qquad 
+
\frac{1}{q_{\mu \nu}}\biggl[-(e^{\mathbf{Q}(t_{s+1}-t_{s})})^{-2}_{y_{s},y_{s+1}}\left(e^{\mathbf{C}^{(\alpha \beta)}_{\eta}(t_{s+1}-t_{s})} \right)_{y_{s}, h+y_{s+1}}(e^{\mathbf{C}^{(\mu \nu)}_{\gamma}(t_{s+1}-t_{s})})_{y_{s}, h+y_{s+1}} 
\\ 
& 
\qquad\qquad \qquad 
+ 
(e^{\mathbf{Q}(t_{s+1}-t_{s})})^{-1}_{y_{s},y_{s+1}}\left(e^{\mathbf{C}_{\psi}^{(\alpha \beta,\mu \nu)}(t_{s+1}-t_{s})} \right)_{y_{s}, 3h+y_{s+1}} \biggr] 
\\
&
\qquad 
-\biggl[-(e^{\mathbf{Q}(t_{s+1}-t_{s})})^{-2}_{y_{s},y_{s+1}}\left(e^{\mathbf{C}^{(\alpha \beta)}_{\eta}(t_{s+1}-t_{s})} \right)_{y_{s}, h+y_{s+1}}(e^{\mathbf{C}^{(\mu )}_{\phi}(t_{s+1}-t_{s})})_{y_{s}, h+y_{s+1}} 
\\
& 
\qquad \qquad \qquad 
+ (e^{\mathbf{Q}(t_{s+1}-t_{s})})^{-1}_{y_{s},y_{s+1}}\left(e^{\mathbf{C}_{\omega}^{(\alpha \beta,\mu)}(t_{s+1}-t_{s})} \right)_{y_{s}, 3h+y_{s+1}}\biggr] \, .
\end{align*}
A similar transform to \eqref{Eq:ExpectedHoldandJumpsTPM} can be applied here to obtain the Hessian from TPMs. When using the result to estimate the error, some knowledge of the number of companies per rating is required.
\subsubsection{Computation of the Error}

Due to the Hessian only being defined for parameters $q_{\alpha \beta}>0$, some parameters are not included in the Hessian, thus the matrix is smaller than $(h-1)^{2}$-\emph{by}-$(h-1)^{2}$. We compute the Hessian as follows,
\begin{itemize}
	\item Let $N_{a}$ be the number of allowed points in the estimated $\mathbf{Q}$ returned from the EM algorithm.
	\item Define an $N_{a}$-\emph{by}-$2$ matrix $\mathbf{V}_\mathbf{Q}$ as the matrix which records the allowed (in the sense of Definition \ref{def:AllowedPairs}) components of $\mathbf{Q}$. The $ij^{th}$ component of the Hessian is then the differential,
	\begin{equation*}
	\frac{\partial^{2}}{\partial q_{\mathbf{V}_{\mathbf{Q}}(i,1) \mathbf{V}_{\mathbf{Q}}(i,2)} \partial q_{\mathbf{V}_{\mathbf{Q}}(j,1)\mathbf{V}_{\mathbf{Q}}(j,2)}} \, . \notag
	\end{equation*}
	\item If we denote the Hessian by the $N_{a}$-\emph{by}-$N_a$ matrix $\mathbf{H}(\cdot)$, then the information matrix is given by $-\mathbf{H}(\cdot)$. The estimated variance of the allowed parameter $q_{ab}$ is then the $i^{th}$ diagonal element of  $-\mathbf{H}(\cdot)^{-1}$, where $\mathbf{V}_{\mathbf{Q}}(i,1)=a$ and $\mathbf{V}_{\mathbf{Q}}(i,2)=b$.
	\item Following standard statistics, the normal based $95\%$ confidence interval for $q_{ab}$ is $q_{ab} \pm 1.96\sqrt{Var(q_{ab})}$.
\end{itemize}


\section{Competitor Algorithms}
\label{sec:CompetitorAlgorithms}

\subsection{Deterministic algorithms}
\label{sec:DeterministicAlgorithms}
\subsubsection*{Diagonal Adjustment (DA)}
The first method to discuss is diagonal adjustment, see \cite{IsraelRosenthalWei2001}. Given a TPM, $\mathbf{P}$, one calculates the matrix logarithm directly. However, due to the embeddability problem, the logarithm may not be a valid generator. To solve this problem \cite{IsraelRosenthalWei2001} suggest setting for $i \neq j$,
\begin{align*}
q^{DA}_{ij} =
\begin{cases}
(\log(\mathbf{P}))_{ij} \, , \, \, \, &\text{if }  (\log(\mathbf{P}))_{ij} \geq 0 \, , \\
0 \, ,   &\text{otherwise.}
\end{cases}
\end{align*}
and adjusting (re-balancing) the diagonal element correspondingly, $q^{DA}_{ii}= \sum_{j \neq i} -q_{ij}$ for $i\in\{1,\cdots,h\}$. 
Hence forcing the corresponding matrix $\mathbf{Q}^{DA}$ to satisfy the properties of a generator.

\subsubsection*{Weighted Adjustment (WA)}
Weighted adjustment is also suggested in \cite{IsraelRosenthalWei2001}. It follows diagonal adjustment except, one re-balances across the entire row. Again, calculate the logarithm of the TPM to find $q$'s, then compute
\begin{align*}
G_{i}=|q_{ii}|+\sum_{j\neq i}\max(q_{ij},0) \, ,\hspace{4mm} B_{i}=\sum_{j \neq i} \max(-q_{ij},0) \, .
\end{align*}
The entries corresponding to weighted adjustment are defined as,
\begin{align*}
q^{WA}_{ij}=
\begin{cases}
0 \, \, \, \, &\text{ if } i \neq j \text{ and } q_{ij}<0\, , \\
q_{ij}-B_{i}|q_{ij}|/G_i   &\text{ otherwise if } G_i>0 \, , \\
q_{ij}  & \text{ otherwise if } G_{i}=0 \, .
\end{cases}
\end{align*}

\subsubsection*{Quasi-Optimisation of the Generator (QOG)}
The above two methods are unfortunately not optimal in any sense. The QOG (Quasi-Optimisation of the Generator), method suggested in \cite{KreininSidelnikova2001} relies on optimisation and is therefore an improvement on the diagonal and weighted adjustment methods. QOG involve solves the minimization problem 
$\min_{\mathbf{Q} \in \mathcal{Q}}\| \mathbf{Q} - \log(\mathbf{P})\|$,
where $\mathcal{Q}$ is the space of stable generator matrices and $|| \cdot ||$ is the Euclidean norm. Further, \cite{KreininSidelnikova2001} provide an efficient algorithm to obtain $\mathbf{Q}$.

\subsection{Statistical algorithm: Markov Chain Monte Carlo}
An alternative statistical algorithm one can adopt is MCMC (Markov Chain Monte Carlo). For the reader's convenience we have included a summary of MCMC in Appendix \ref{Sec:MCMC Overview}. It should be noted that all MCMC algorithms presented here use a so-called auxiliary variable technique, by introducing the fully observed Markov chain, $X$ as a random variable. Moreover, the prior for $\textbf{Q}$ is $\Gamma(\alpha,1/\beta)$ (shape and scale), which is conjugate for the likelihood of a CTMC.

\subsubsection{Gibbs Sampling - Bladt \& Sorensen 2005}
\label{Sec:BS05algorithm}

To simulate the Markov process, $X$, \cite{BladtSorensen2005} suggest a rejection sampling method. As is stated in \cite{BladtSorensen2005}, such a sampling method runs into difficulties when considering low probability events since the rejection rate will be high (e.g. default for high rated bonds). The MCMC algorithm is summarised as follows, \cite{Inamura2006},
 \begin{enumerate}
 	\item Construct an initial generator $\mathbf{Q}$ using the prior distribution ($\Gamma(\alpha_{ij},1/\beta_{i})$).
 	\item For some specified number of runs
 	\begin{enumerate}[(1)]
 		\item Simulate $X$ for each observation from $Y$, with law according to $\mathbf{Q}$.
 		\item Calculate the quantities of interest $K$ and $S$, from $X$.
 		\item Construct a new $\mathbf{Q}$ by drawing samples from $\Gamma(K_{ij}(t)+\alpha_{ij}, 1/(S_{i}(t)+ \beta_{i}))$.
 		\item Save this $\mathbf{Q}$ and use it in the next simulation.
 	\end{enumerate}
 	\item From the list of $\mathbf{Q}$s, drop some proportion (burn in), then take the mean of the remainder.
 \end{enumerate}
 The issues with this method are the choice of $\alpha$ and $\beta$ and the number of runs required before we know that the sample has converged (burn in). Both of these are critical in obtaining accurate answers from MCMC and although \cite{BladtSorensen2005} suggested taking $\alpha_{ij}$ and $\beta_i$ to be 1, they observe MCMC overestimating entries in the generator when true entries were small. Furthermore, here we are required to use the TPM indirectly through inferring company transitions. That is, we consider $M$ companies in each rating and define the number of companies to make the transition $i$ to $j$ as $M\times P_{ij}$, this of course need not be an integer, but we can always normalise the entries. The reason we cannot use the TPM directly as we did in the EM algorithm is due to the fact that MCMC becomes very sensitive to the values in the prior. The burn in for MCMC will be of little concern to us here as will become apparent when carrying out analysis on the algorithms.

\subsubsection{Importance Sampling - Bladt \& Sorensen 2009}
\label{Sec:BS09algorithm}

  \cite{BladtSorensen2009} address some of the issues in \cite{BladtSorensen2005} by running the same algorithm as previous combined with an importance sampling scheme based on the Metropolis-Hastings algorithm (in its essence a single component Metropolis-Hastings algorithm). The proposal distribution suggested is a Markov chain with generator given by the `neutral matrix' $\mathbf{Q}^{*}$, which takes the following form,
 \begin{align*}
 	\mathbf{Q}^{*} 
 		=\frac{1}{W}\big(\mathbf{1}_{h}-\textbf{I}_{h} -h \textbf{I}_{h}\big) \, , 		
 	\end{align*}
 where $\mathbf{1}_{h}$ and $\textbf{I}_{h}$ is the $h$-by-$h$ matrix of ones and identity matrix respectively and $W$ is a scaling factor set to match the intensities in the true generator matrix $\mathbf{Q}$. \cite{BladtSorensen2009} note that if entries in $\mathbf{Q}$ are known to be zero, then the corresponding element in $\mathbf{Q}^{*}$ should also be set to zero and the diagonal modified accordingly. Thus transitions rarely produced by the generated Markov chain will occur much more frequently under $\mathbf{Q}^{*}$. Thus we have solved (at least partially) one of the problems faced in MCMC.
 The importance sampling weights for a chain $X$ are,
 \begin{align*}
 w(X)= \frac{L(\mathbf{Q};X)}{L(\mathbf{Q}^{*};X)} \, ,
 \end{align*}
where $L$ is the CTMC likelihood.
 For the priors, \cite{BladtSorensen2009} do not suggest any significant improvement on their earlier work. The authors use $\alpha=1$ and $\beta=5$, which they claim gives better results than the suggestion in \cite{BladtSorensen2005}. However, it still provides a problem when dealing with entries in $\mathbf{Q}$ which are close to zero. The problem stems from the fact that very little information is known (rarely observed) for certain transitions, therefore the output for these entries is mostly based on our prior beliefs.
 
\subsubsection{MCMC Mode Algorithm}
\label{sec:MCMCMODEalgorithmInamura2006}

\cite{Inamura2006} presented an alternative algorithm to the original MCMC algorithm presented in \cite{BladtSorensen2005} whereby one calculates the mode rather than the mean. The author claims that this gives extremely accurate results and outperforms other algorithms. The reasoning presented is that the standard MCMC overestimates in the small probability cases due to the gamma distribution being `skewed', therefore the mode is a better estimate. \cite{Inamura2006} approximates the mode of $\{q_{ij}^{(k)}\}$ by kernel smoothing over the estimates (after taking the log transform to ensure all results are positive).

\begin{remark}[Other MCMC based estimators] Many extensions and different MCMC methods to solve this problem are possible (e.g. priors as hyperparameters or sequential Monte Carlo techniques). Here, we consider less complex MCMC algorithms which already set the tone for a comparative study.
\end{remark}



\section{Benchmarking the Algorithms} 
\label{Sec:Comparing Algorithms}

Due to the diversity of investments bank's make, one cannot assess an algorithms' performance with a single test. With this in mind we consider a host of tests on different portfolios and matrices. The computations were carried out on a Dell PowerEdge R430 with four Intel Xeon E5-2680 processors. During the review process of our work we found \cite{Pfeuffer2017} with an R-implementation of some of the algorithms covered in the previous section. The performance tests of \cite{Pfeuffer2017} are a just subset of those we present next and independently confirm (where there is overlap) our findings, in particular the timing of the MCMC algorithms versus the EM. A version of our algorithms will appear in the mentioned R-package (see Remark \ref{rem:pfeuffer-ctmcd}).

The first observation we make is, transition matrices can vary substantially depending on the financial climate (see \cite{ChristensenEtAl2004} and \cite{Cantor2004}). Therefore we consider two different generator matrices which can be thought of as the generator in financial stress and the generator in financial calm. In order to keep these matrices `reasonable' we start off with the generator given in \cite{ChristensenEtAl2004} built using a large amount of data (see also \cite{Inamura2006}) and consider a generator which has in general higher transition rates and one with lower transition rates. Through considering more than one generator this provides a more detailed assessment of the performance of the various algorithms than other comparative reviews, such as \cite{Inamura2006}. The generators we consider are shown in Table~\ref{Table:Unstable Generator} and Table~\ref{Table:Stable Generator}. We observe that Table~\ref{Table:Unstable Generator}, has more non-zero entries and larger entries than that of Table~\ref{Table:Stable Generator}. 

\begin{table}[!h]
\centering
	\footnotesize
	\begin{tabular}{ c | c | c | c | c | c | c | c | c |}
		~& AAA & AA & A & BBB & BB & B & C & D \\ \hline
		AAA & -0.146371 & 0.085881 & 0.04549 & 0.015 & 0 & 0 & 0 & 0 \\ \hline
		AA & 0.018506 & -0.166337 & 0.114831 & 0.033 & 0 & 0 & 0 & 0 \\ \hline
		A & 0.0276 & 0.047012 & -0.198043 & 0.09043 & 0.023001 & 0.01 & 0 &0 \\ \hline
		BBB & 0.011469 & 0.010734 & 0.088133 & -0.243046 & 0.077569 & 0.044407 & 0.010734 & 0 \\ \hline
		BB & 0 & 0 & 0.019159 & 0.184699 & -0.323077 & 0.106166 & 0.013053 & 0 \\ \hline
		B & 0 & 0 & 0.012280 & 0.034822 & 0.093489 & -0.296265 & 0.134273 & 0.022401 \\ \hline
		C & 0 & 0 & 0 & 0 & 0.02 & 0.140209 & -0.600939 & 0.440730 \\ \hline
		D & 0 & 0 & 0 & 0 & 0 & 0 & 0 & 0 
		\\ \hline
	\end{tabular}
	\caption{True unstable generator}
	\label{Table:Unstable Generator}
\end{table}
\begin{table} [!h]
\centering
	\footnotesize
	\begin{tabular}{ c | c | c | c | c | c | c | c | c |}
		~& AAA & AA & A & BBB & BB & B & C & D \\ \hline
		AAA & -0.061371 & 0.055881 & 0.005490 & 0 & 0 & 0 & 0 & 0 \\ \hline
		AA & 0.013506 & -0.096337 & 0.074831 & 0.008 & 0 & 0 & 0 & 0 \\ \hline
		A & 0 & 0.037012 & -0.097442 & 0.06043 & 0 & 0 & 0 &0 \\ \hline
		BBB & 0 & 0.000734 & 0.058133 & -0.120843 & 0.057569 & 0.004407 & 0 & 0 \\ \hline
		BB & 0 & 0 & 0.009159 & 0.104699 & -0.190024 & 0.076166 & 0 & 0 \\ \hline
		B & 0 & 0 & 0 & 0.024822 & 0.083489 & -0.174985 & 0.064273 & 0.002401 \\ \hline
		C & 0 & 0 & 0 & 0 & 0 & 0.080209 & -0.300939 & 0.220730 \\ \hline
		D & 0 & 0 & 0 & 0 & 0 & 0 & 0 & 0 \\ \hline
	\end{tabular}
	\caption{True stable generator}
	\label{Table:Stable Generator}
\end{table}

Throughout the analysis we refer to the multiple MCMC algorithms introduced in Section \ref{sec:CompetitorAlgorithms} which we label in the following way: \verb|MCMC BS05| is \cite{BladtSorensen2005}'s algorithm of Section \ref{Sec:BS05algorithm}; \verb|MCMC BS09| is \cite{BladtSorensen2009}'s algorithm of Section \ref{Sec:BS09algorithm}; and \verb|MCMC Mode| is \cite{Inamura2006}'s algorithm in Section \ref{sec:MCMCMODEalgorithmInamura2006}.

\subsection{Sample Size Inference}
The first test we consider is an extension to a test in \cite{Inamura2006}, where the author considers a true underlying generator and masks it by using it to simulate TPMs, which we view as observations, then applying the algorithms to each observation. The key point here is, \cite{Inamura2006} only simulates $100$ companies per rating and hence the outputted TPM is non-embeddable (has $0$ entries for accessible jumps). This is an extremely useful test because it provides a fair and intuitive way to assess the performance of each algorithm, however, \cite{Inamura2006} only considers one true generator and only one level of information i.e. $100$ companies per rating. Alongside the two different generators we also consider a range of companies per rating to determine its effect on convergence for each algorithm. Furthermore, \cite{Inamura2006} uses seven years worth of data, although one would likely have access to multiple years worth of TPM data, it is unlikely that we would have seven years of transitions from the same generator. Hence we consider four years, which is more consistent with time homogeneity estimates for generators (see \cite{ChristensenEtAl2004}).
We calculate our estimates for the generator as follows.
\begin{enumerate}
\setlength\itemsep{0em}
	\item Take a range of obligors per rating, $[100,200, 300, 500, 750, 1000]$ and $10$ random seeds.
	\item For each true generator simulate four one year TPMs for each seed and for each obligor per rating. Hence we have (\#Years$\times$\#Obligors categories$\times$\#Random Seeds$\times$\#True generators), simulated TPMs.
	\item For each set of four simulated TPM we estimate the generator for each algorithm. MCMC may take a long time to run, therefore we consider the time taken to carry out the first 10 runs and the total time taken, if these exceed 180 or 18 000 seconds respectively, the algorithm is deemed to be too slow and no result is returned. Note, MCMC algorithms use 3000 runs with a burn in of 300. This is smaller than \cite{Inamura2006} for example, however, \cite{Inamura2006} shows apparent convergence to the stationary distribution in a small number of iterations and we observe a similar result.
	\item Therefore, for each algorithm we have (\# Obligors categories $\times$ \# Random Seeds $\times$ \# True generators) estimated generators to analyze.
\end{enumerate} 
We analyze the estimated generators by considering, distance between estimated generator and true generator in Euclidean norm and difference in one year probability of default. All results presented have been obtained by analyzing the estimated generator for each seed, then averaging. This gives a better picture of the average performance.
\begin{table}[htb]
\centering
		\small
		\begin{tabular}{ c | c | c | c}
Algorithms & Deterministic & EM  & MCMC\\ \hline
Time (seconds)  & $< 1$ & $\sim 10$ &$ \sim 10^3$ to $\sim 10^4$\\ \hline
		\end{tabular}
		\caption{Order of time taken to execute the various algorithms. Note that MCMC also depends on the level of information i.e. obligors in each rating. We also note that BS 09 algorithm is faster than the other MCMC algorithms but still takes $10^4$ seconds in the case of 1000 obligors per rating.}
		\label{Table:Running Times}
\end{table}

\subsubsection{Convergence in Euclidean Norm}
Our goal in this analysis is to consider the empirical rate of improvement of each algorithm as our `information' about the true generator increases. For each obligor category we calculate the natural log of the distance (measured by the Euclidean norm) between the estimate and the true. The results are shown in figures ~\ref{EuclideanAlgorithmConvergenceUnstable}~ and~ \ref{EuclideanAlgorithmConvergenceStable}.

\begin{figure}[!ht] 
	\centering
	\includegraphics[width=13.5cm]{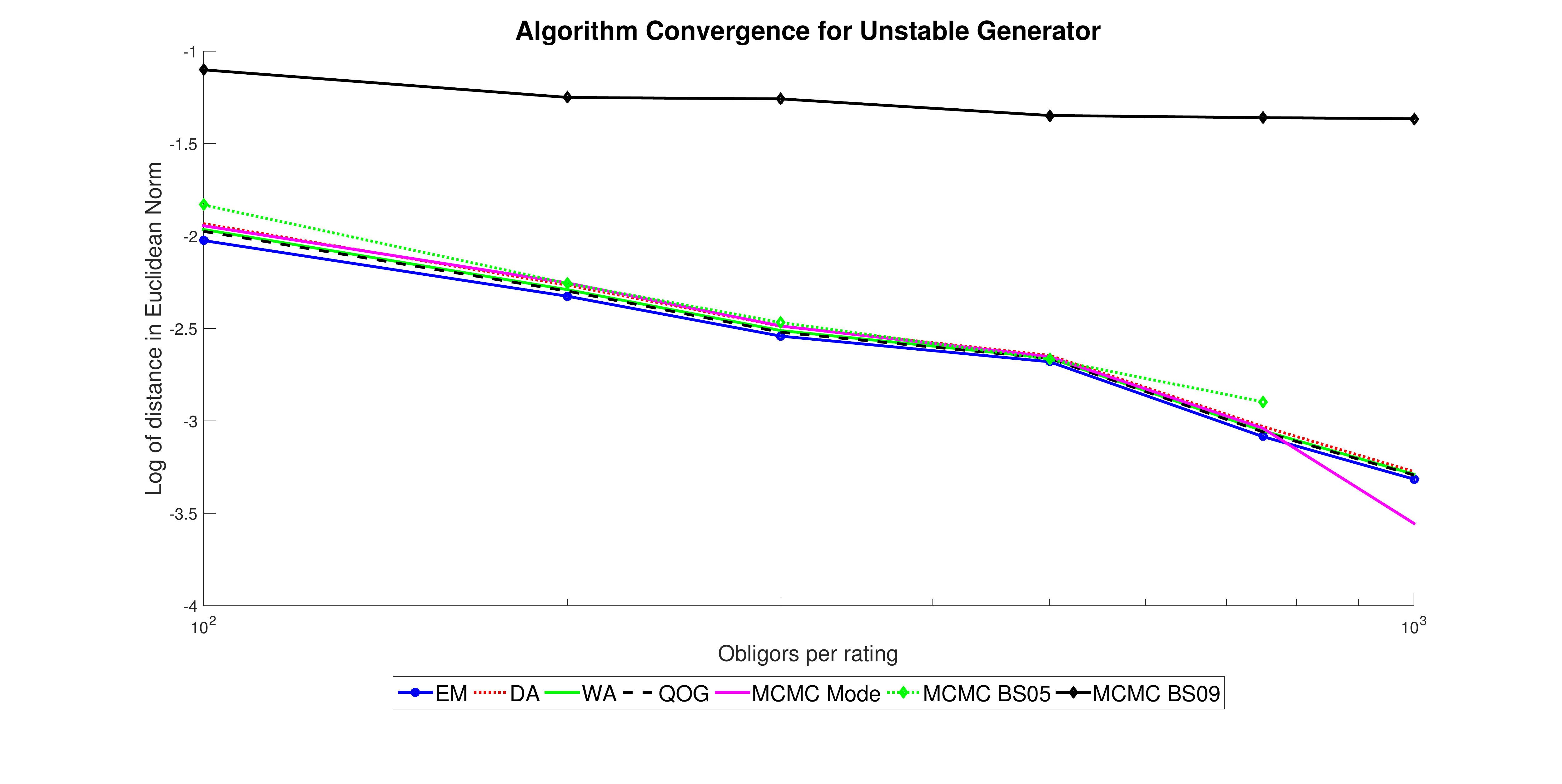}
	\vspace{-1cm}
	\caption{Showing the log of the error for each algorithm as a function of obligors per rating.}
	\label{EuclideanAlgorithmConvergenceUnstable}
\end{figure} 

\begin{figure}[!ht]
	\centering
	\includegraphics[width=13.5cm]{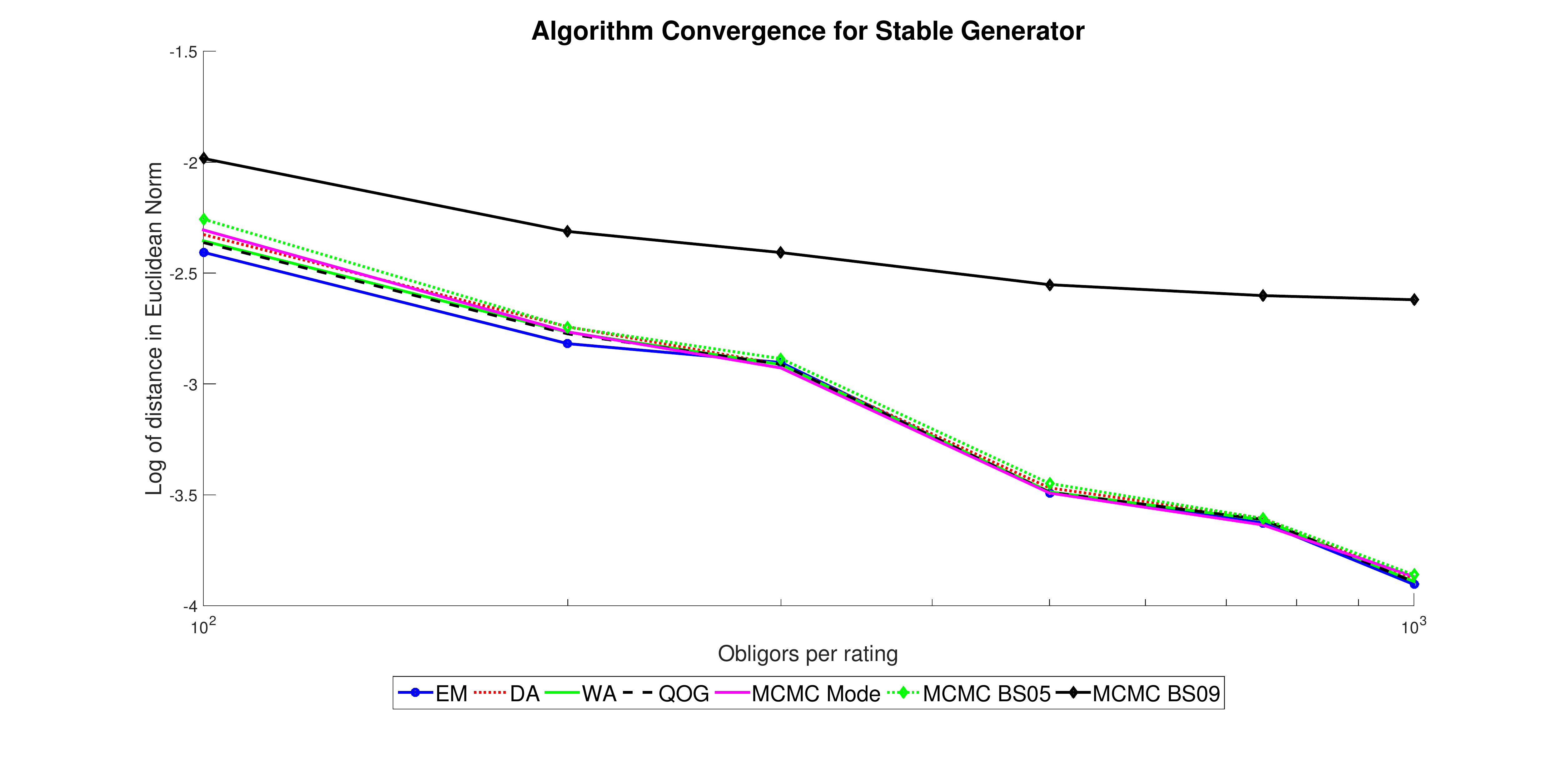}
	\vspace{-1cm}
	\caption{Showing the log of the error for each algorithm as a function of obligors per rating.}
	\label{EuclideanAlgorithmConvergenceStable}
\end{figure} 

Note the $x$-axis is on a logarithmic scale. We observe similarities between the two figures, most notably in the case of low information all algorithms have very similar convergence results, however, as we increase the information there is substantial variation in improvement, \verb|MCMC BS09| algorithm does not improve as well as the other algorithms. Missing points stem from an algorithm failing the acceptance times. 

The MCMC algorithms have a potentially increased error due to the Monte Carlo simulation, lowering it requires a larger computational expense to the  already most expensive algorithm being tested here. For the \cite{BladtSorensen2009} algorithm, the neutral matrix approximation may give poor mixing, thus the additional error.


\subsubsection{Error in Probability of Default}
Although overall error is important, it does not provide details on the small probability scale. This is extremely important in banking, since estimation of the probability of default is crucial. Using the same estimated generators as previous we calculate the corresponding one year TPM, that is, we calculate $\exp\{\mathbf{Q}_{\text{estimate}}\}$ (using the {\tt {expm}} function in MATLAB) for each seed then take the average. The averaged TPM default probabilities are compared to the true ones. To keep the numbers in the comparisons meaningful we plot the log of the relative error, where we define,
\begin{equation*}
\text{Relative Error}=\frac{|\text{PD}_{\text{estimate}}-\text{PD}_{\text{true}}|}{\text{PD}_{\text{true}}} \, .
\end{equation*} 
The results of which are given in Figures \ref{RelativeDefaultErrorUnstableAAABBB} and  \ref{RelativeDefaultErrorUnstableBBCCC}. 

\begin{figure}[!ht] 
	\centering
	\includegraphics[width=17cm]{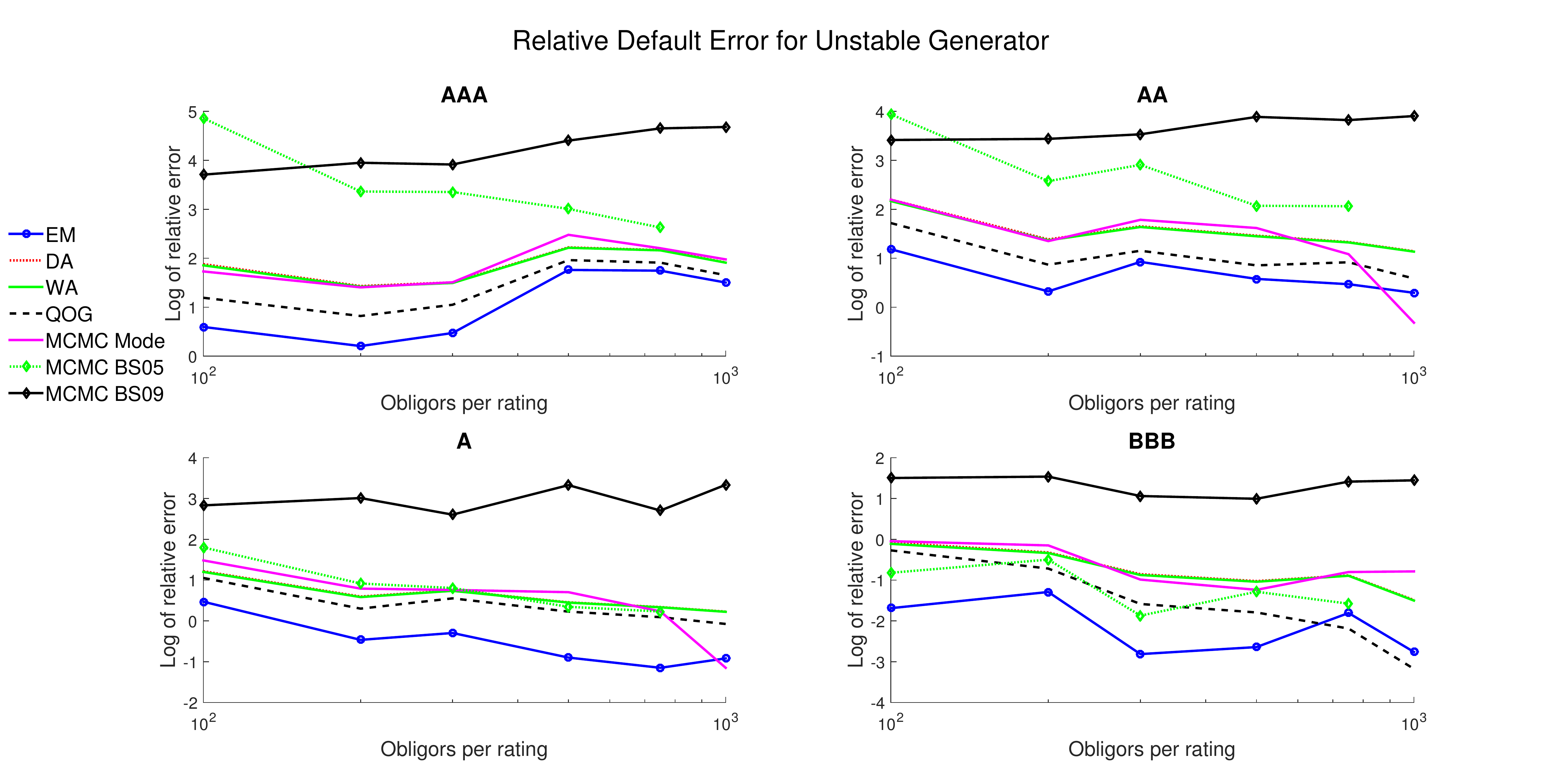}
	\vspace{-1cm}
	\caption{Showing the log of the relative default error for each algorithm as a function of obligors per rating.}
	\label{RelativeDefaultErrorUnstableAAABBB}
\end{figure} 

\begin{figure}[!ht] 
	\centering
	\includegraphics[width=17cm]{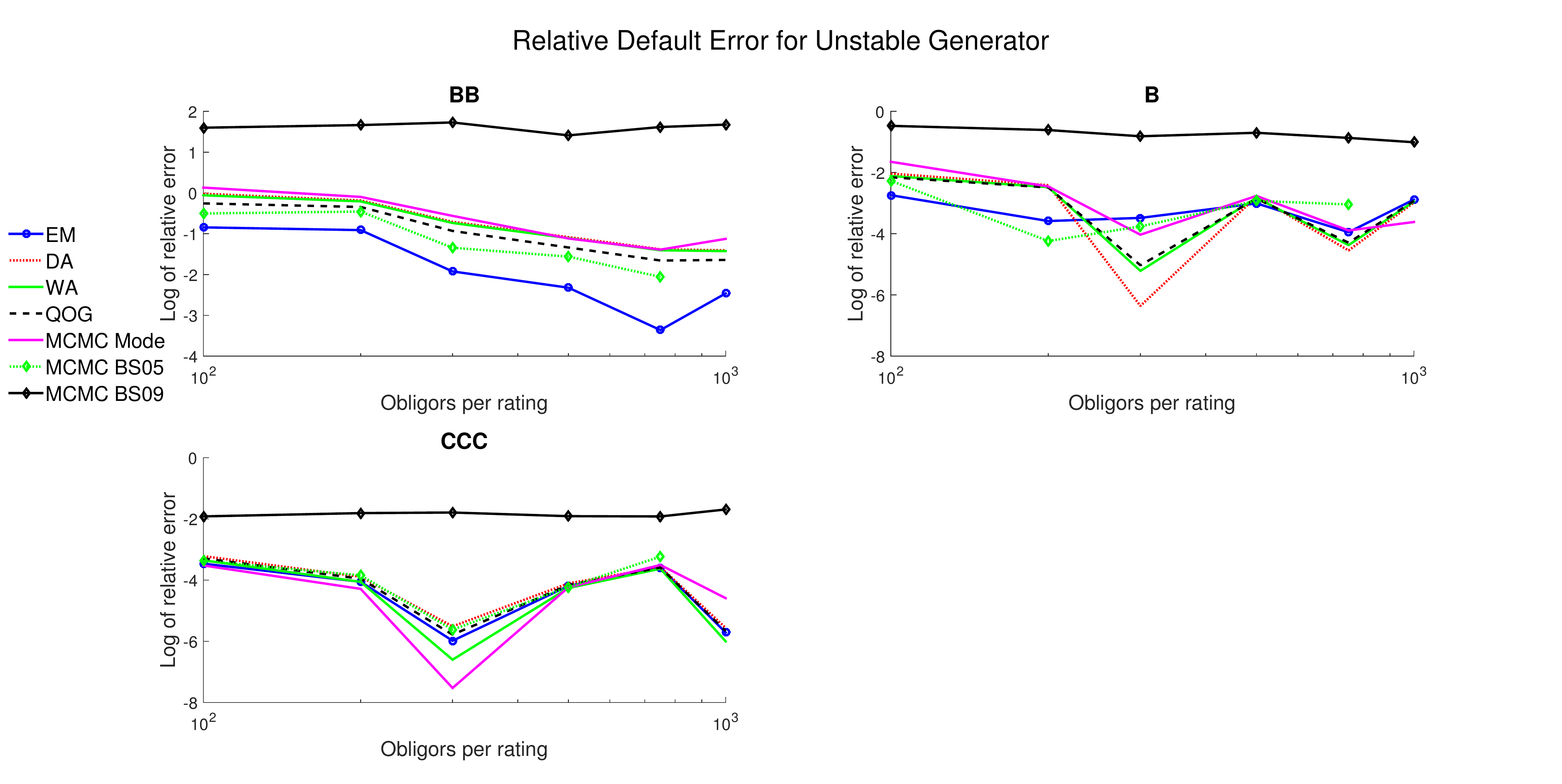}
	\vspace{-1cm}
	\caption{Showing the log of the relative default error for each algorithm as a function of obligors per rating.}
	\label{RelativeDefaultErrorUnstableBBCCC}
\end{figure} 

%

Unlike the overall error, there appears to be far greater volatility in the error estimation w.r.t. the probability of default. Moreover, there appears to be no general downward trend in error for the investment grade ratings. A likely cause for this is, even with 1000 companies there are still no/few investment grade defaults. Of the algorithms \verb|MCMC BS09| performs the worst. The EM algorithm though has consistently one of the smallest errors and is clearly the best in the investment grades.We have only shown the results for the unstable generator, the stable generator was similar.


\subsection{Time Dependent Probability of Default}
A key question that has not been addressed in the literature is how do the probabilities of default change in time among the several algorithms. For ths we only consider EM, QOG, WA and the MCMC Mode algorithm from \cite{Inamura2006}, since these algorithms gave the best probability of default estimates. 

We consider a non-embeddable TPM, then estimate the generator matrix $\mathbf{Q}$, from $\mathbf{Q}$ we can easily calculate the probability of a company with some initial rating defaulting in time $t>0$. The goal here is to assess how that probability changes with time. The TPM is given in Table \ref{Tab:TPM for Default Prob}, for the MCMC algorithm we took this table to be generated with 250 obligors per rating. 

\begin{table} [H]
	\centering
			\footnotesize
		\begin{tabular}{ c | c | c | c | c | c | c | c | c |}	
			~& AAA & AA & A & BBB & BB & B & C & D \\ \hline
			AAA & 0.8824 &  0.1176 & 0 & 0 & 0 & 0 & 0 & 0 \\ \hline
			AA & 0.0064 & 0.9111 & 0.0813 & 0.0008 & 0.0001 & 0 & 0.0003 & 0 \\ \hline
			A & 0.0003 & 0.0559 & 0.8836 & 0.0499 & 0.0079 & 0.0015 & 0.0002 & 0.0007 \\ \hline
			BBB & 0 & 0.0116 & 0.1585 & 0.7640 & 0.0528 & 0.0070 & 0 & 0.0061 \\ \hline
			BB & 0 & 0 & 0.0213 & 0.1193 & 0.7746 & 0.0623 & 0.0099  & 0.0127 \\ \hline
			B & 0 & 0 & 0.0062 & 0.0199 & 0.1669 & 0.7017 & 0.0730 & 0.0322 \\ \hline
			C & 0 & 0 & 0 & 0 & 0.0417 & 0.2083 & 0.4544 & 0.2956 \\ \hline
			D & 0 & 0 & 0 & 0 & 0 & 0 & 0 & 1 \\ \hline
		\end{tabular}
	\caption{Observed TPM used to estimate the generators in probability of default plots.}
	\label{Tab:TPM for Default Prob} 
\end{table}

\begin{figure}[!ht] 
	\centering
	\includegraphics[width=17cm]{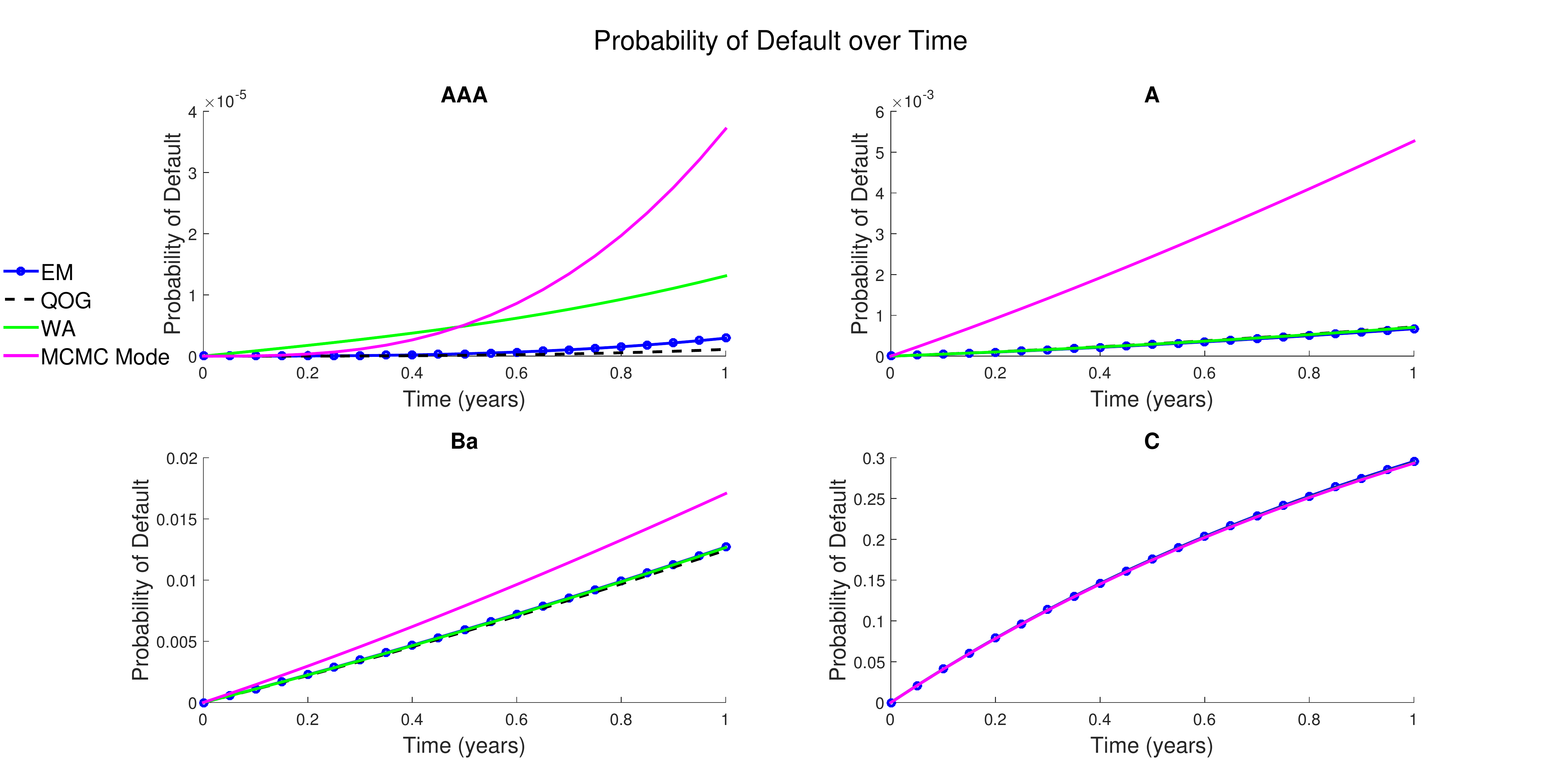}
	\vspace{-1cm}
	\caption{Probability of default over time for EM, QOG, MCMC Mode and WA.}
	\label{Fig:Default probs in time}
\end{figure} 

The probability of default across ratings over the one year time horizon is found in Figure \ref{Fig:Default probs in time}. The plots give a deeper understanding to the algorithms themselves. As the probability of default increases the algorithms converge, however, in the case of less defaults we observe a much larger discrepancy. This can be thought of as the algorithm's ability to deal with missing data, in the lower grades we observe defaults and thus have a handle on the probability, however, in the case of AAA ratings we observe no defaults and therefore it is an approximation by the algorithm. This shows the difference between the methods, shows the potential prior dependence in the MCMC algorithm. What is also extremely interesting is that QOG set the jump in the generator from AA to C as zero (even though the TPM has a non zero entry there), this implies QOG may in some places under estimate the risk for the investment ratings, this can be seen by the fact QOG puts a smaller probability of default on AAA. 

There is a clear ovestimation of the probability of default at higher grades by the WA and MCMC algorithms.

\subsection{Risk Charge}
The previous tests have been rather theoretical, we now consider a practical test to asses the performance of these algorithms in calculating risk charges. We do not give much discussion to the calculation of these risk charges for more technical details readers should consult texts such as \cite{SkoglundChen2011}. Here we consider multiple stylized portfolios to represent the risk appetites of different banks. To best of our knowledge analysis into how different risk measures react to different portfolio types has not been considered in the literature. The risk charges we consider are IRC (VaR at 99.9\% with a 3 months liquidity horizon including mark to market loss), IDR (VaR at 99.9\% over one year only considering default) and a theoretical risk charge which is IRC but measured using Expected Shortfall (ES) at 97.5\%. The final risk charge is included due to the Basel committee showing an increasing interest in ES. We consider 4 years worth of simulated data, and to keep the analysis realistic we consider 200 companies per rating. We consider 3 different portfolios corresponding to risk adverse (all investment grade), a speculative portfolio (all speculative grades) and finally a mixed portfolio. The portfolios considered are given in Tables \ref{Table:Mixed Portfolio}, \ref{Table:Investment Portfolio} and \ref{Table:Speculative Portfolio}. The tables show the values and ratings of the various bonds in each portfolio.

\begin{table}[hbt] 
	\parbox{.3\linewidth}{
		\centering
			\footnotesize
		\begin{tabular}{| l | l |}
			\hline
			AAA &  100, 500, 1500, 750 \\ \hline
			AA & 200, 750, 2000, 650 \\ \hline
			A & 150, 400, 400\\ \hline
			BBB & 300, 500, 150, 1500 \\ \hline
			BB & 500, 250, 700 \\ \hline
			B &  200, 500\\ \hline
			C & 100, 150, 200\\ \hline
		\end{tabular}
		\caption{Mixed portfolio}
		\label{Table:Mixed Portfolio}
	}
	\hfill
	\hspace{-2cm}
	\parbox{.35\linewidth}{
		\centering
			\footnotesize
		\begin{tabular}{| l | l |}
			\hline
			AAA &  1000, 500,\\
			    & 1500, 1500 \\ \hline
			AA & 100, 400, 750, \\
			   & 2000, 400, 1500 \\ \hline
			A & 150, 100, 800,\\
			     & 400, 200\\ \hline
			BBB &  \\ \hline
			BB &  \\ \hline
			B &  \\ \hline
			C & \\ \hline
		\end{tabular}
		\caption{Investment portfolio}
		\label{Table:Investment Portfolio}
	}
	\hfill
	\hspace{-2.0cm}
	\parbox{.35\linewidth}{
		\centering
			\footnotesize
		\begin{tabular}{| l | l |}
			\hline
			AAA &   \\ \hline
			AA & \\ \hline
			A &\\ \hline
			BBB &  \\ \hline
			BB & 1000, 150, 100,\\
			  & 800, 1500 \\ \hline
			B &  100, 300, 400, \\
			  &  750, 2000, 1500\\ \hline
			C & 400, 500, 400, 1000\\ \hline
		\end{tabular}
		\caption{Speculative portfolio}
		\label{Table:Speculative Portfolio}
	}
\end{table}

Alongside these portfolios we calculate the risk charges using the following information,
\begin{itemize}
	\item The interest rates we receive for a bond in each rating are
	\begin{center}
		\small
		\begin{tabular}{ c | c |c  | c | c | c | c }
			AAA & AA & A & BBB & BB & B & C \\ \hline
			2.65\% & 2.69\% & 2.78\% & 2.93\% & 3.18\% & 5.45\% & 12.39\% 
		\end{tabular}
	\end{center}
	These figures are based on interest rates from Moody's and can be found in Section 4.1 of \cite{SkoglundChen2011}. Although these interest rates do not technically match the generators we are using for the TPMs they provide reasonable interest rates for our toy example.
	\item We assume that all money is lost in the case of default (zero recovery rate).
	\item We calculate credit migration using the one factor\footnote{This is technically not the true regulation for the calculation of IDR which requires a two factor model, however our goal here is only to use these calculations as a method for comparing algorithms.} credit metrics model (\cite{GuptonEtAl1997}), i.e. normalised asset returns follow,
	\begin{equation*}
	z_{i}= \beta_{i}X + \sqrt{1-\beta_{i}^2}\epsilon_i \, ,
	\end{equation*}
	where $X$ is the systematic risk, $\epsilon_i$ is the idiosyncratic risk both standard normally distributed and $\beta_i$ is the correlation to the systematic risk, defined in \cite{Basel2003}*{p.50},
	\begin{equation*}
	\beta_{i}= 0.12 \left( \frac{1-\exp\{-50\, P^{D}_{i}\}}{1-\exp\{-50\}} \right)+0.24\left(1-\frac{1-\exp\{-50\, P^{D}_{i}\}}{1-\exp\{-50\}}\right) \, ,
	\end{equation*}
	where $P^{D}_{i}$ is the probability of default of asset $i$. Consequently we see that the higher $P_{i}^{D}$ the lower the value of $\beta$. 
	\item Although more sophisticated methods are available for calculation of VaR and ES (see \cite{Fermanian2014}), we calculate the risk charges using Monte Carlo. This is sufficient here since the portfolios are small relative to a typical bank portfolio, therefore we can obtain accurate estimates using a reasonable number of simulations. 
	\item Again, we calculate 10 realizations of the TPMs and estimate a generator for each.
\end{itemize}
We consider 15$\times 10^{5}$ simulations for each portfolio, to assess whether this was sufficient we calculated VaR and ES using 7.5$\times 10^{5}$, 10$\times 10^{5}$, 12.5$\times 10^{5}$ and 15$\times 10^{5}$ simulations and found the difference between 7.5$\times 10^{5}$ and 15$\times 10^{5}$ to be $<5\%$ for all cases. Hence were are confident that 15$\times 10^{5}$ gives sufficiently accurate results for our purposes.

With respect to the risk charge calculation, similar to the previous analysis, we calculate the risk charges for every set of TPMs, then average over all the seeds to obtain the risk charge. The risk charges as set by the true generators are given in Table \ref{Table: Risk charge true}. 
%
%
%

\begin{table}[!hbt] 
	\begin{center}
			\footnotesize
		\begin{tabular}{ c | c |  c | c | c | c | c }
~ & \multicolumn{1}{c}{} & \multicolumn{1}{c}{Stable} & ~ & \multicolumn{1}{c}{} & \multicolumn{1}{c}{Unstable} & ~ \\ \hline
			~ & Mixed & Investment & Speculative & Mixed & Investment & Speculative\\ \hline
			IRC & 702 & 0.32 & 3395 &1251 & 0.41 & 5057 \\ \hline

			IRC ES & 508 & 0.20 & 2409  & 842 & 3.78 & 3826 \\ \hline
			
			IDR & 750 & 0 & 3400 & 1750 & 200 & 4600 \\ \hline
		\end{tabular}
		\caption{Risk charge results for the true generators.}
		\label{Table: Risk charge true}
	\end{center}
\end{table}

To asses the performance of each algorithm we measure the error by the following,
\begin{equation*}
\text{Risk Error} = \frac{\frac{1}{N}\sum_{i=1}^{N} |\text{Risk Charge Estimate}(i) - \text{Risk Charge True}|}{\text{Risk Charge True}} \, ,
\end{equation*}
where $\text{Risk Charge Estimate}(i)$ is the $i^{th}$ realization of the risk charge and $N$ is the number of TPM sets (10 here). The results obtained by the algorithms are shown in Table \ref{Table:Risk Charge Algorithm Results}.

\begin{table}[!ht] 
	\begin{center}
			\footnotesize
		\begin{tabular}{c| c | c |  c | c | c | c | c }
~ &~ & \multicolumn{1}{c}{} & \multicolumn{1}{c}{Stable} & ~ & \multicolumn{1}{c}{} & \multicolumn{1}{c}{Unstable} & ~ \\ \hline
			~&~ & Mixed & Investment & Speculative & Mixed & Investment & Speculative\\ \hline

			~ & EM & 7.3 & 7.5 & 1.5 & 22.5 & 29 195 & 2.6 \\ \cline{2-8}
			~ & DA & 11.9 & 8.1 & 2.4 & 36.9 & 66 829 & 4.3 \\ \cline{2-8}
			~ & WA & 11.8 & 8.1 & 2.3 & 37.3 & 69 293 & 4.1 \\ \cline{2-8}
			\multirow{2}{*}{\textbf{IRC}} & QOG & 11.6 & 7.8 & 2.3 & 26.7 & 38 976 & 4.1 \\ \cline{2-8}
			~ & MCMCBS05 & 154 & 306 000 & 2 & 49.6 & 478 000 & 4.1 \\ \cline{2-8}
			~ & MCMCBS09 & 24.9 & 18.4 & 14.4 & 68.3 & 264 000 & 14 \\ \cline{2-8}
			~ & MCMCMode & 12.5 & 8.1 & 3.6 & 34.9 & 39 000 & 3.9 \\ \hline

			~ & ~ & \multicolumn{1}{c}{} & \multicolumn{1}{c}{~} & ~ & \multicolumn{1}{c}{} & \multicolumn{1}{c}{~} & ~ \\ \hline
			~ & EM & 5.3 & 115 & 3.4 & 8.6 & 375 & 2.7 \\ \cline{2-8}
			~ & DA & 8.2 & 235 & 5.1 & 16.6 & 1130 & 3.9 \\ \cline{2-8}
			~ & WA & 7.8 & 210 & 5 & 16.4 & 1109 & 3.8 \\ \cline{2-8}
			\multirow{2}{*}{\textbf{IRC ES}} & QOG & 7.3 & 123 & 4.9 & 12.6 & 622 & 3.8 \\ \cline{2-8}
			~ & MCMCBS05 & 35.4 & 135 000 & 4.7 & 19.7 & 5315 & 4.1 \\ \cline{2-8}
			~ & MCMCBS09 & 21 & 610 & 15.5 & 67.7 & 6693 & 13.1 \\ \cline{2-8}
			~ & MCMCMode & 9.2 & 235 & 6.1 & 19.1 & 1063 & 3.5 \\ \hline

			~ & ~ & \multicolumn{1}{c}{} & \multicolumn{1}{c}{~} & ~ & \multicolumn{1}{c}{} & \multicolumn{1}{c}{~} & ~ \\ \hline
			~ & EM & 6 & 0 & 0.3 & 4.3 & 113 & 3.5 \\ \cline{2-8}
			~ & DA & 10 & 0 & 1.2 & 8.6 & 295 & 5.7 \\ \cline{2-8}
			~ & WA & 9.3 & 0 & 0.6 & 8.6 & 295 & 5.2 \\ \cline{2-8}
			\multirow{2}{*}{\textbf{IDR}} & QOG & 7.3 & 0 & 0.6 & 5.4 & 185 & 5.3 \\ \cline{2-8}
			~ & MCMCBS05 & 139 & 1580 & 0.3 & 12.6 & 530 & 4.7 \\ \cline{2-8}
			~ & MCMCBS09 & 20 & 40 & 9.3 & 33.7 & 775 & 13.2 \\ \cline{2-8}
			~ & MCMCMode & 10 & 10 & 0.9 & 8 & 278 & 4.7 \\ \hline
		\end{tabular}
		\caption{Risk charge results for each algorithm as a \%.}
		\label{Table:Risk Charge Algorithm Results}
	\end{center}
\end{table}

It should be noted, in the stable IDR some algorithms produce a non-zero value for the investment portfolio, therefore we have inserted the money value. The first observation we make is, all algorithms overestimate the risk for the investment portfolio. This is down to two key feature, one is the `step like' nature of VaR, where in a small portfolio, small probability changes can make a large difference. The other is because we are averaging over multiple Monte Carlo simulations, thus having one default in one of those realizations will change the overall average dramatically. In terms of a typical bank portfolio this type of error should not be a problem since we would be dealing with a far larger number of assets and hence one would obtain multiple defaults. However, the results do still give a useful comparison between the algorithms. Although the MCMC algorithms can outperform the deterministic algorithms for the speculative grades, remarkably in all categories the EM produces the best results. From the tests we have considered we conclude the EM to be the superior algorithm for this problem.



\subsection{Error estimation of the EM algorithm}
\label{section:EMerrorEstimation}
A major advantage of the statistical algorithms over their deterministic counterparts is that one can derive error estimates (confidence intervals) without the brute force (slightly ad-hoc) method of bootstrapping. For MCMC this comes by looking at the posterior distribution, which we get for free. However, as we have seen MCMC is computationally expensive and we have derived a relatively cheap formula to calculate the confidence intervals from the EM algorithm.

In a similar fashion to the analysis we have carried out previously we now test the error estimate given by the EM. Again, we mask the true generator by using simulated TPMs, however, here we only consider the scenario of 300 obligors per rating, but the number of years worth of data is varied. That is, we simulate 50 years worth of TPMs and then apply the EM algorithm using 1 years worth then 2 years etc up to 50 years. This analysis shows both the estimated error for the parameter and also how the error changes when more information is added. It should also be noted that we replace companies who have defaulted to the rating they were pre default. This keeps the number of companies in the system constant and can be thought of as the flow of new companies being rated. Moreover, this is only one realisation of the data, hence the parameter estimate and confidence intervals are not particularly smooth.

\begin{figure}[!ht] 
	\centering
	\includegraphics[width=16cm]{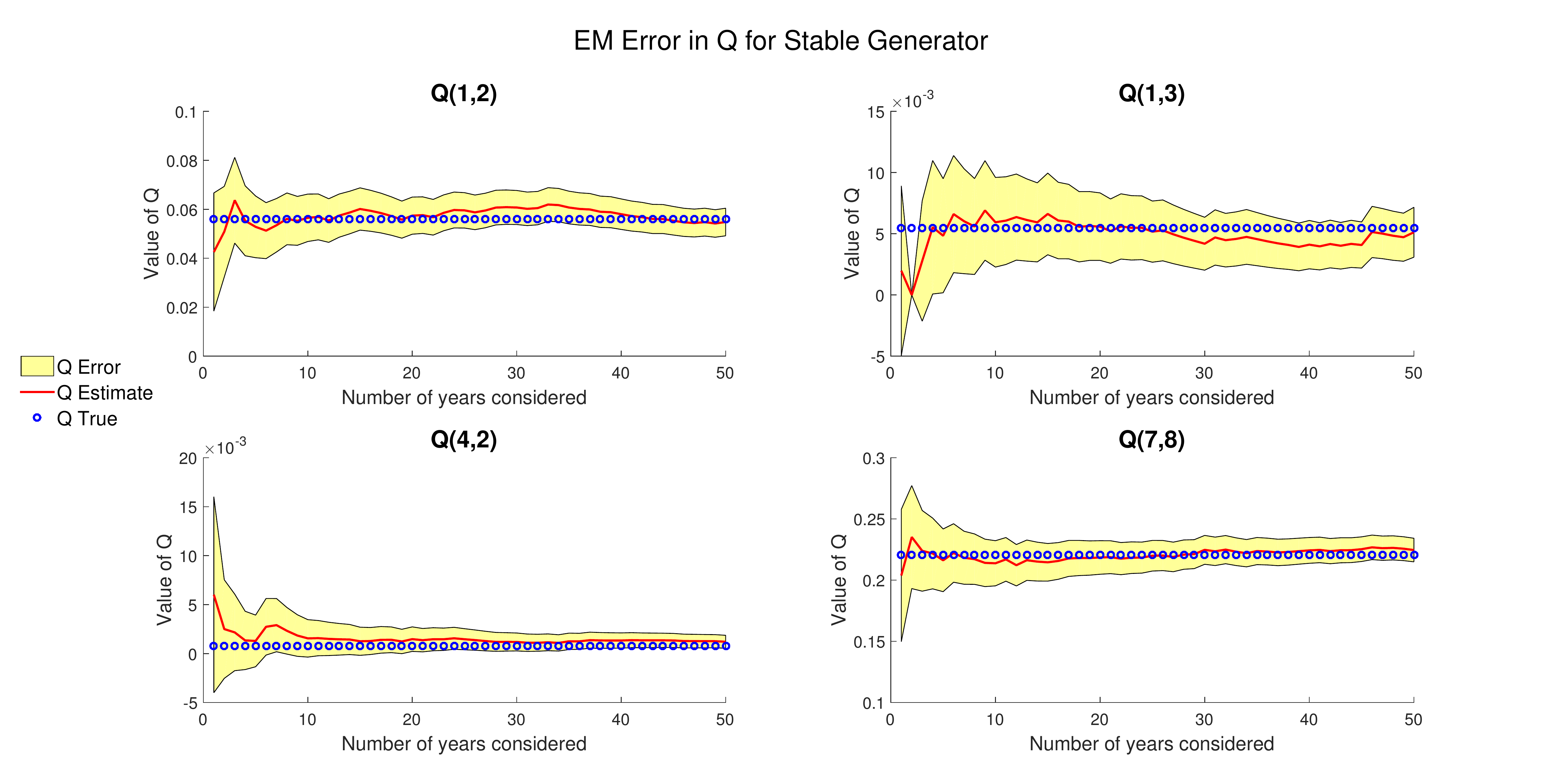}
	\vspace{-0.5cm}
	\caption{Showing the estimated 95\% confidence interval for parameters as a function of years.}
	\label{EMErrorForStable}
\end{figure} 

\begin{figure}[!ht] 
	\centering
	\includegraphics[width=16cm]{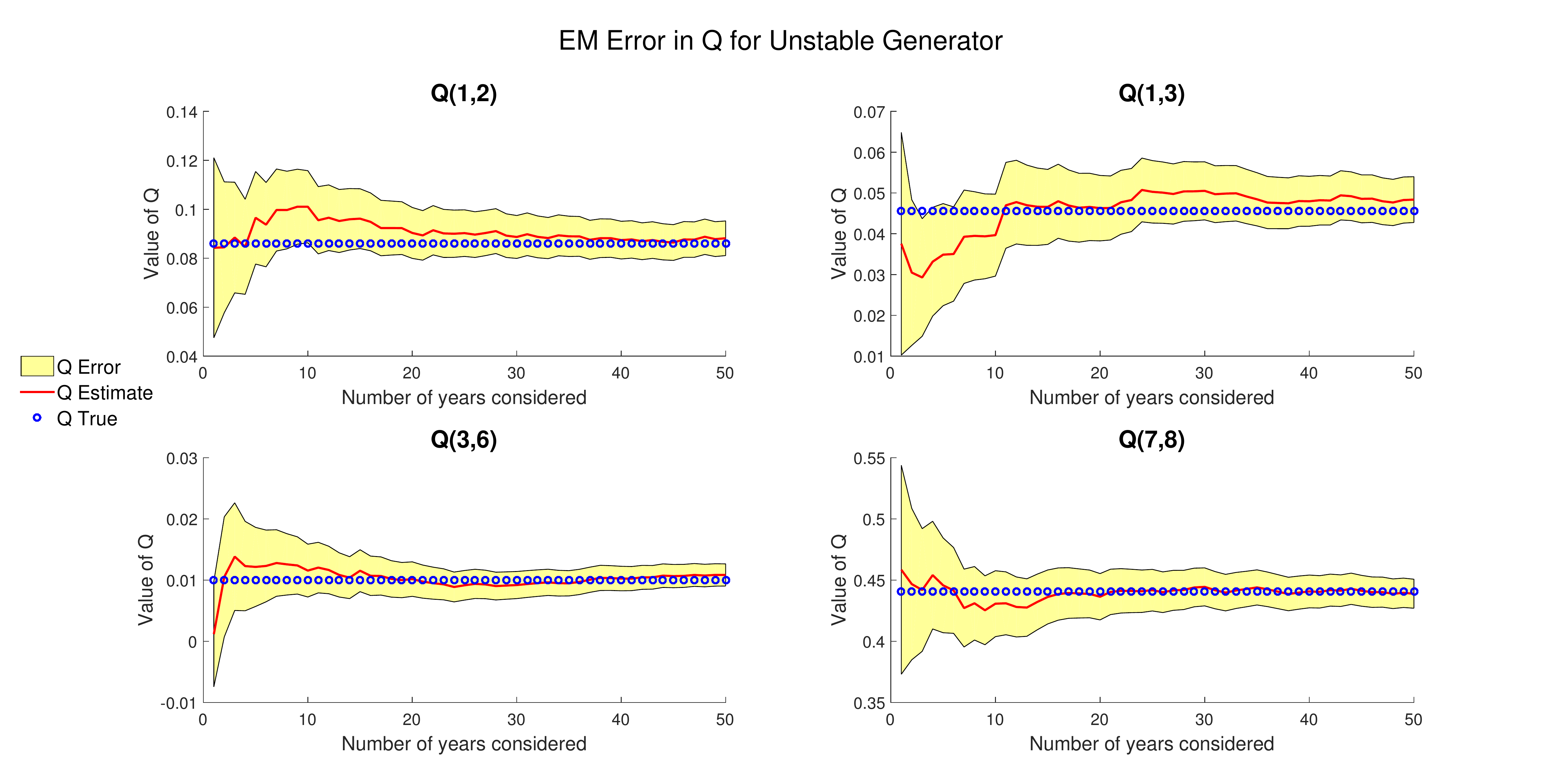}
	\vspace{-0.5cm}
	\caption{Showing the estimated 95\% confidence interval for parameters as a function of years.}
	\label{EMErrorForUnstable}
\end{figure} 


The transitions shown in Figures \ref{EMErrorForStable} and \ref{EMErrorForUnstable} were chosen to show a spectrum of the magnitudes in the generators, the other entries not shown are similar. The first point to make is that the true value of the parameter is almost always within the confidence interval and the confidence interval shrinks as the number of years increases.
One of most important features though is that the confidence interval is only small when the EM is stable and close to the true parameter, hence the EM is not ``over confident'' but is also providing reasonable estimates on its own error. The final point to make is, although some confidence intervals go below zero by a small amount, this is only true in the case where the parameter is extremely close to zero initially and further, once more data is considered, all parameters have confidence interval which are strictly positive.
\begin{remark}
[Confidence intervals in practice and regulation]
Figures \ref{EMErrorForStable} and \ref{EMErrorForUnstable} show a very important feature. Namely, how much the estimate can vary with data, especially when the parameter is reasonably small. Such analysis exposes the variance in the estimate and how much data is actually required before the estimate levels out. From this example though the confidence intervals calculated from the information matrix appear to be able to capture this error. In the view of future regulation it may be prudent to take such confidence intervals into account when considering risk charges. 
\end{remark}

\subsubsection*{Connection to the Global Maximum} 
\label{Sec:Connection to Global Max}
A previous problem with the EM was one could not be sure of the nature of the stationary point. However, we know the form of the Hessian, and therefore we can easily check if this point is a maximum by assessing the eigenvalues of this matrix. Clearly, if we were not at a maximum, then it would be worth perturbing the outputted generator and rerunning the algorithm. As discussed in Remark \ref{rem:LikelihoodUniqueMaximizer} the question of a global maximum is very difficult in this setting.
\begin{remark}
One way that has been suggested to improve the chances of the EM converging to the global maximum is, to start from multiple points. Here we can consider creating starting points by setting for each $i \neq j$, $q_{ij} \sim \text{Exp}(\lambda)$ for an appropriate $\lambda$ then setting $q_{ii}$ appropriately.
\end{remark}
We tested the EM according to the above remark and found in every case considered the EM always returns the same generator.


\section{Conclusions and future research}

In this manuscript we built upon the closed form expressions for the expected number of jumps and holding times of a CTMC with an absorbing state, over given observations and used the results to 
derive a closed form expression for the Hessian of the likelihood. This coupled with stronger convergence has elevated the EM algorithm to be the optimal algorithm to tackle this problem.


Across the battery of tests carried out, the EM algorithm outperforms competing algorithms. The EM is a tractable algorithm, slower than the deterministic algorithms but still several orders of magnitude faster than the Markov-Chain Monte-Carlo alternatives (Table \ref{EuclideanAlgorithmConvergenceStable}). The statistical algorithms (EM and MCMC) embed a strong robustness property for the estimator contrary to the deterministic algorithms, i.e. the likelihood is far less sensitive to small changes in the underlying TPM. In terms of estimating risk charges, the EM algorithm has superior results in all scenarios.

On the more practical side, Figure \ref{Fig:Default probs in time} highlights that for lower ratings algorithms produce essentially the same estimates for the probabilities of default while a palpable difference emerges at higher ratings. Moreover, the error estimates in the EM may provide a sensible way to test in effect model risk.

Lastly, non-Markovian phenomena like rating momentum (see \cite{LandoSkodeberg2002}) and appropriate models to tackle it will be addressed in forthcoming research.

\section*{Acknowledgement(s)}
The authors would like to thank Dr.~R.~P.~Jena at Nomura Bank plc London for the helpful comments. In addition, the authors would like to thank Ruth King (U.~of Edinburgh), Ioannis Papastathopoulos (U.~of Edinburgh) and Samuel Cohen (Oxford Uni.) for the helpful discussions.

We thank as well the two anonymous referees for their comments which led to improvements on the initial submission.

\section*{Funding}
\small
G. Smith was supported by The Maxwell Institute Graduate School in Analysis and its
Applications, a Centre for Doctoral Training funded by the UK Engineering and Physical
Sciences Research Council (grant [EP/L016508/01]), the Scottish Funding Council, Heriot-Watt
University and the University of Edinburgh.

G. dos Reis acknowledges support from the \emph{Funda{\c c}$\tilde{\text{a}}$o para a Ci$\hat{e}$ncia e a Tecnologia} (Portuguese Foundation for Science and Technology) through the project [UID/MAT/00297/2013] (Centro de Matem\'atica e Aplica\c c$\tilde{\text{o}}$es CMA/FCT/UNL).

%
%
%
%

\appendix

\section{Proofs}
\subsection{Proof of Lemma \ref{Lem:ExpectationBounds}}
\label{Sec:Proof of Expectation Lemma}
We now provide the proof of Lemma \ref{Lem:ExpectationBounds}, all terms used have the same definition as they did when the Lemma was stated.
Throughout we assume $i \neq h$, thus from from Assumption \ref{Assump:Element Constraint} $\bP_{\mathbf{Q}}(X(t)=j|X(0)=i)>0$ for all $j\in \{1, \dots, h\}$ and $t>0$. The first inequality we prove is the lower bound on the expected number of jumps. Following the assumptions in Lemma \ref{Lem:ExpectationBounds} and time homogeneity we make the observation
\begin{equation*}
\bE_{\mathbf{Q}}[K_{ij}(T)|\mathbf{P}] \ge P_{ij}^{u}\bP_{\mathbf{Q}}(K_{ij}(t) \ge 1 |X(0)=i, X(t)=j) \, .
\end{equation*}
The above inequality holds because we are only considering $X(0)=i$, $X(t)=j$ and not all possible combinations of start and end states, moreover, $\bP_{\mathbf{Q}}(K_{ij} \ge 1|X(0)=i, X(t)=j) \le \sum_{n=1}^{\infty}n\bP_{\mathbf{Q}}(K_{ij}=n|X(0)=i, X(t)=j)$.
We further observe,
\begin{equation*}
\bP_{\mathbf{Q}}(K_{ij} \ge 1|X(0)=i, X(t)=j) \ge \frac{q_{ij}}{-q_{ii}} \, .
\end{equation*}
Thus the lower bound in inequality \eqref{Eq:Bounds on expected jumps} can be easily obtained. 
We now prove the upper bound on the expected number of jumps. The first observation we make is for all $ \nu \in\{1, \dots, h\}$,
\begin{align*}
\bE_{\mathbf{Q}}[K_{ij}(T)|X(0)=i, X(t)=\nu] = \sup_{\mu \in \{1, \dots, h\}} \bE_{\mathbf{Q}}[K_{ij}(T)|X(0)=\mu , X(t)=\nu] \, .
\end{align*}
To see this, let $\mu \neq i$, then denote by $\tau_{i}$ the first time the process enters state $i$ (if $\bP_{\mathbf{Q}}(X(t)=i|X(0)=\mu)=0$ for $t>0$, then the result is trivial), by the law of total probability we find,
\begin{align*}
& \bE_{\mathbf{Q}}[K_{ij}(t)|X(0)=\mu , X(t)=\nu] \\
&\qquad
= \bE_{\mathbf{Q}}[K_{ij}(t)|X(0)=\mu , X(t)=\nu, \tau_{i} <t] \bP_{\mathbf{Q}}(\tau_{i} <t|X(0)=\mu , X(t)=\nu) 
\\
&
\qquad \qquad 
+\bE_{\mathbf{Q}}[K_{ij}(t)|X(0)=\mu , X(t)=\nu, \tau_{i}  \ge t] \bP_{\mathbf{Q}}(\tau_{i} \ge t|X(0)=\mu , X(t)=\nu)\, .
\end{align*}
The second term is zero. Then, using the Markov property we obtain,
\begin{align*}
\bE_{\mathbf{Q}}[K_{ij}(t)|X(0)=\mu , X(t)=\nu] & \le \bE_{\mathbf{Q}}[K_{ij}(t)|X(\tau_{i})=i , X(t)=\nu, \tau_{i} <t] 
\\
&\le \bE_{\mathbf{Q}}[K_{ij}(t)|X(0)=i , X(t)=\nu] \, .
\end{align*}
Consequently from this observation and \eqref{Eq:ExpectedHoldandJumpsTPM} we obtain,
\begin{align*}
\bE_{\mathbf{Q}}[K_{ij}(T)|\mathbf{P}] \le h N \sum_{\nu =1}^{h} \bE_{\mathbf{Q}}[K_{ij}(t)|X(0)=i , X(t)=\nu] \, .
\end{align*}
Observe that,
\begin{align*}
\bE_{\mathbf{Q}}[K_{ij}(t)|X(0)=i , X(t)=\nu] = \dfrac{\bE_{\mathbf{Q}}[K_{ij}(t) \1_{\{X(t)=\nu\}}|X(0)=i ]}{\bP_{\mathbf{Q}}(X(t)=\nu|X(0)=i)} \le \dfrac{\bE_{\mathbf{Q}}[K_{ij}(t)|X(0)=i ]}{\bP_{\mathbf{Q}}(X(t)=\nu|X(0)=i)} \, .
\end{align*}
The numerator is easy to bound by considering the expected number of jumps out of $i$,
\begin{align*}
\bE_{\mathbf{Q}}[K_{ij}(t)|X(0)=i ] \le -q_{ii}t \, .
\end{align*}
The denominator requires further analysis, firstly, let $n=|i - \nu|$, and therefore by Assumption \ref{Assump:Element Constraint} we can go from state $i$ to $\nu$ in $n$ jumps, w.l.o.g. let $i \ge \nu$ (it will be come clear that the ordering does not matter). Firstly, if $i= \nu$ then,
\begin{align*}
\bP_{\mathbf{Q}}(X(t)=\nu|X(0)=i) \ge e^{q_{ii}t} \, .
\end{align*}
For $i> \nu$, we use the Markov property to obtain,
\begin{align*}
\bP_{\mathbf{Q}}(X(t)=\nu|X(0)=i) \ge \prod_{a=1}^{n} \bP_{\mathbf{Q}}\left(X\left(\frac{a}{n}t\right)=i+a \biggl|X\left(\frac{a-1}{n}t\right)=i+a-1\right) \, .
\end{align*}
Conditioning on $X$ only making one jump in each increment we obtain,
\begin{align*}
\bP_{\mathbf{Q}}(X(t)=\nu|X(0)=i) 
&\ge
\prod_{a=1}^{n} \frac{q_{i+a-1, i+a}}{-q_{i+a-1, i+a-1}}(-q_{i+a-1, i+a-1})t \exp\{q_{i+a-1, i+a-1}t\} 
\\
&
\ge
\prod_{a=1}^{n} \epsilon t \exp\{-h t / \epsilon\} \, .
\end{align*}
As $n \le h$ and the terms are strictly smaller than $1$, the sought result follows (independent of $\nu \neq i$).

The last inequality to prove concerns the holding times. By taking for $P^{u}_{ii}>0$,
\begin{align*}
\bE_{\mathbf{Q}}[S_{i}(T)|\mathbf{P}]  \ge P^{u}_{ii}  \bE_{\mathbf{Q}}[S_{i}(t)|X(0)=i, X(t)=i]  \ge P^{u}_{ii} t \exp\{q_{ii}t\} \, ,
\end{align*}
where the final inequality follows by simply considering the case of no jumps. We can then apply the bounds from Assumption \ref{Assump:Element Constraint} to complete the inequality.

\subsection{Proof of Theorem \ref{theo:SecondDerivativesR}}
\label{Sec:Proof of second derivatives}
We recall from \cite{Wilcox1967}, \cite{TsaiChan2003} that for a square matrix $\mathbf{M}$ whose elements depend on a vector of parameters $\{\lambda_{1}, \dots \lambda_{r}\}$ (for $r \in \bN$), the following identity holds
\begin{equation}
\label{Eq:MatrixParameterDifferentiation}
\frac{\partial e^{\mathbf{M}({\mathbf{ \lambda}})t}}{\partial \lambda_{i}}= \int_{0}^{t} e^{(t-u)\mathbf{M}({\mathbf{ \lambda}})} \frac{\partial \mathbf{M}({\mathbf{ \lambda}})}{\partial \lambda_{i}} e^{u \mathbf{M}({\mathbf{ \lambda}})} du \, , 
\end{equation}
for all $i \in \{1,\dots,r\}$. Let $\mu,\nu,\alpha,\beta \in \{1,\dots,h\}$. Recalling Proposition \ref{Prop:Expected Values}, differentiating $\bE_{\mathbf{Q}}[K_{\mu \nu}(t)|y]$ w.r.t. $q_{\alpha \beta}$ yields,
\begin{align*}
\dfrac{\partial}{\partial q_{\alpha \beta}}\bE_{\mathbf{Q}}[K_{\mu \nu}(t)|y]
=
& 
\sum_{s=1}^{n-1} 
                -(e^{\mathbf{Q}(t_{s+1}-t_{s})})^{-2}_{y_{s},y_{s+1}}
								\left(\dfrac{\partial}{\partial q_{\alpha \beta}}e^{\mathbf{Q}(t_{s+1}-t_{s})} \right)_{y_{s},y_{s+1}}
								(e^{\mathbf{C}^{(\mu \nu)}_{\gamma}(t_{s+1}-t_{s})})_{y_{s}, h+y_{s+1}} 
\\
& 
\qquad 
\,              + (e^{\mathbf{Q}(t_{s+1}-t_{s})})^{-1}_{y_{s},y_{s+1}}
                       \left(\dfrac{\partial}{\partial q_{\alpha \beta}}
											          e^{\mathbf{C}^{(\mu \nu)}_{\gamma}(t_{s+1}-t_{s})} 
											\right)_{y_{s}, h+y_{s+1}} \, .
\end{align*}
Note that although the expected value of $K$ only depends on individual elements of the matrix and not the full matrix, we are still able to use the differentiation result since $A_{ij} = \mathbf{e}_{i}^{\intercal}\mathbf{A} \mathbf{e}_{j}$. Hence, from \eqref{Eq:MatrixParameterDifferentiation} we obtain,
\begin{align*}
\dfrac{\partial}{\partial q_{\alpha \beta}}\bE_{\mathbf{Q}}[K_{\mu \nu}(t)|y]=
& 
\sum_{s=1}^{n-1} 
-(e^{\mathbf{Q}(t_{s+1}-t_{s})})^{-2}_{y_{s},y_{s+1}}
\left(\int_{0}^{t} e^{(t-u)\mathbf{Q}} \frac{\partial \mathbf{Q}}{\partial q_{\alpha \beta}} e^{u\mathbf{Q}} du \right)_{y_{s},y_{s+1}}
(e^{\mathbf{C}^{(\mu \nu)}_{\gamma}(t_{s+1}-t_{s})})_{y_{s}, h+y_{s+1}}  
	\\
	&
	\qquad \qquad 
	+ (e^{\mathbf{Q}(t_{s+1}-t_{s})})^{-1}_{y_{s},y_{s+1}}
	\left(\int_{0}^{t} e^{(t-u)\mathbf{C}_{\gamma}^{(\mu \nu)}} 
	    \frac{\partial \mathbf{C}_{\gamma}^{(\mu \nu)}}{\partial q_{\alpha \beta}} 
			 e^{u\mathbf{C}_{\gamma}^{(\mu \nu)}} du \right)_{y_{s}, h+y_{s+1}} \, .
\end{align*}
Clearly, since $q_{\alpha \beta}$ appears twice in $\mathbf{Q}$,
\begin{equation}
\frac{\partial \mathbf{Q}}{\partial q_{\alpha \beta}} = \mathbf{e}_{\alpha}\mathbf{e}_{\beta}^{\intercal} -\mathbf{e}_{\alpha}\mathbf{e}_{\alpha}^{\intercal} \, , \qquad \text{and} \qquad 
\frac{\partial \mathbf{C}_{\gamma}^{(\mu \nu)}}{\partial q_{\alpha \beta}} =\left[ {\begin{array}{cc}
	\mathbf{e}_{\alpha}\mathbf{e}_{\beta}^{\intercal} -\mathbf{e}_{\alpha}\mathbf{e}_{\alpha}^{\intercal} & \mathbf{e}_{\mu}\mathbf{e}_{\nu}^{\intercal} \delta_{\mu \alpha} \delta_{\nu \beta} \\
	0 & \mathbf{e}_{\alpha}\mathbf{e}_{\beta}^{\intercal} -\mathbf{e}_{\alpha}\mathbf{e}_{\alpha}^{\intercal}
	\end{array} } \right] \, . \notag
\end{equation}
Then, by \cite{VanLoan1978} 
we can solve these integrals explicitly to obtain,
\begin{align*}
\dfrac{\partial}{\partial q_{\alpha \beta}}\bE_{\mathbf{Q}}[K_{\mu \nu}(t)|y]=& \sum_{s=1}^{n-1} -(e^{\mathbf{Q}(t_{s+1}-t_{s})})^{-2}_{y_{s},y_{s+1}}\left(e^{\mathbf{C}^{(\alpha \beta)}_{\eta}(t_{s+1}-t_{s})} \right)_{y_{s}, h+y_{s+1}}(e^{\mathbf{C}^{(\mu \nu)}_{\gamma}(t_{s+1}-t_{s})})_{y_{s}, h+y_{s+1}} 
	\\
	&
	\qquad \qquad 
	+(e^{\mathbf{Q}(t_{s+1}-t_{s})})^{-1}_{y_{s},y_{s+1}}\left(e^{\mathbf{C}_{\psi}^{(\alpha \beta,\mu \nu)}(t_{s+1}-t_{s})} \right)_{y_{s}, 3h+y_{s+1}} \, ,
\end{align*}
again $\mathbf{C}^{(\alpha \beta)}_{\eta}$ and $\mathbf{C}_{\psi}^{(\alpha \beta,\mu \nu)}$ are as defined in the Theorem's statement.

Therefore, we have a closed form expression for the derivative of expected jumps w.r.t. $q_{\alpha \beta}$. Applying a similar argument for the expected holding time we obtain,
\begin{align*}
\dfrac{\partial}{\partial q_{\alpha \beta}}\bE_{\mathbf{Q}}[S_{\mu}(t)|y]
=
& \sum_{s=1}^{n-1} -(e^{\mathbf{Q}(t_{s+1}-t_{s})})^{-2}_{y_{s},y_{s+1}}\left(e^{\mathbf{C}^{(\alpha \beta)}_{\eta}(t_{s+1}-t_{s})} \right)_{y_{s}, h+y_{s+1}}(e^{\mathbf{C}^{(\mu )}_{\phi}(t_{s+1}-t_{s})})_{y_{s}, h+y_{s+1}} 
\\
& \, 
\qquad 
+(e^{\mathbf{Q}(t_{s+1}-t_{s})})^{-1}_{y_{s},y_{s+1}}\left(e^{\mathbf{C}_{\omega}^{(\alpha \beta,\mu)}(t_{s+1}-t_{s})} \right)_{y_{s}, 3h+y_{s+1}} \, ,
\end{align*}
where $\mathbf{C}_{\omega}^{(\alpha \beta,\mu)}$ is as defined in the Theorem.
Combining these yields the required result.


\section{Overview of Markov Chain Monte Carlo algorithm}
\label{Sec:MCMC Overview}
For details on the Markov Chain Monte Carlo (MCMC) theory we refer the reader to \cite{GilksRichardsonEtAl1996}. Algorithms for implementing MCMC to estimate a generator, from discrete observations are discussed in \cite{BladtSorensen2005} and \cite{BladtSorensen2009}. MCMC differs from the EM in the sense that EM estimates the set of parameters which maximises the likelihood function, while MCMC samples from the posterior distribution. Namely, given some data $D$, the posterior distribution of parameters $\theta$ is $\pi(\theta|D)$, which by Bayes' theorem is,
 \begin{align*}
 \pi(\theta|D)
= 
 \frac{\pi(D|\theta)\pi(\theta)}{\int \pi(D|\theta)\pi(\theta)d\theta} \, ,
 \end{align*}
with $\pi(D|\theta)$ denoting the likelihood and $\pi(\theta)$ the prior distribution. MCMC obtains the best guess of $\theta$ by sampling from $\pi(\theta|D)$ and taking the Monte Carlo approximation of the expected value. The reason the expectation is our best guess is due to the fact we use both the data (likelihood) but also our experience on what the outcome should approximately be (the prior). Although the prior can be extremely useful in stopping `bad' answers it is also a criticism of MCMC due to so-called prior sensitivity. 

\begin{remark}
	Here we purely discuss MCMC to sample from the posterior, algorithms which approximate the maximum likelihood in the presence of missing data do exist, but are more useful when for example one cannot explicitly write the E step in the EM algorithm (see \cite{GelfandCarlin1993}).
\end{remark}

Similar to the case of the EM algorithm the problem faced here is missing data. Namely we wish to consider the so-called posterior distribution of the generator matrix $\mathbf{Q}$, which we denote by $\pi(\mathbf{Q}|D)$ (although it is common to suppress the data and only write $\pi(\mathbf{Q})$). The difficulty is, in its current state this is an extremely hard distribution to evaluate so we augment with an auxiliary variable $X$ (see  \cite{GilksRichardsonEtAl1996}*{p.105} and \cite{BesagGreen1993}). In general $X$ need not require an interpretation, although here it will correspond to the full Markov chain. In order to generate realisations of $\pi(\mathbf{Q}|D)$, we specify the conditional distribution $\pi(X|\mathbf{Q},D)$ which provides the joint distribution $\pi(\mathbf{Q},X|D)=\pi(\mathbf{Q}|D)\pi(X|\mathbf{Q},D)$ and therefore the marginal distribution of $\mathbf{Q}$ is $\pi(\mathbf{Q}|D)$. One can then sample from the marginal distribution by using any sampling method that preserves the joint distribution $\pi(\mathbf{Q},X|D)$ (and by extension $\pi(\mathbf{Q}|D)$), such as Gibbs or Metropolis Hastings.

The method used in \cite{BladtSorensen2005} and \cite{BladtSorensen2009} is the data augmentation algorithm from \cite{TannerWong1987} (see also \cite{LittleRubin2002}*{p.200}). We specify the prior distribution $\pi(\mathbf{Q})$ and take a realisation from this distribution, $\mathbf{Q}^{(0)}$, we then construct a sequence $\{\mathbf{Q}^{(k)},X^{(k)}\}$, for $k=1,\dots, M$ by:
\begin{enumerate}
	\item Draw, $X^{(k)} \sim \pi(X|\mathbf{Q}^{(k-1)},D)$.
	\item Draw, $\mathbf{Q}^{(k)} \sim \pi(\mathbf{Q}|X^{(k)},D)=\pi(\mathbf{Q}|X^{(k)})$ (since $X^{(k)}$ is richer than $D$).
	\item Save $\{\mathbf{Q}^{(k)},X^{(k)}\}$  and take $k=k+1$.
\end{enumerate}
Under mild conditions (see \cite{GilksRichardsonEtAl1996}*{Chapter 4}), after some burn-in $n$, the sequence $\{\mathbf{Q}^{(k)},X^{(k)}\}$ for $k \ge n$ has the same distribution as $\pi(\mathbf{Q},X|D)$. Moreover, the marginals also have the correct distribution, namely, $\{\mathbf{Q}^{(k)}\} \sim \pi(\mathbf{Q}|D)$ for $k \ge n$. Therefore we estimate the generator matrix by, $\frac{1}{M-n+1}\sum_{k=n}^{M} \mathbf{Q}^{(k)}$.

For the choice of prior, $\pi(\mathbf{Q})$, \cite{BladtSorensen2005} suggest a prior from the gamma distribution with shape $\alpha_{ij}$ and scale $1/\beta_{i}$. Hence, $q_{ij} \sim \Gamma(\alpha_{ij},1/\beta_{i})$, where $\alpha_{ij} ,\beta_{i} \ge 0$, $\forall \, i \neq j \in \{1,\dots,h\}$.  With this choice, the prior is a 
 conjugate prior. Although this prior has some drawbacks, we note, by assuming the prior to follow a Gamma distribution we effectively bound the parameter space, therefore there is no  need to make the space compact. Noting that the posterior distribution of $X$ is equivalent to the likelihood i.e. $\pi(X|\mathbf{Q})=L_{t}(X;\mathbf{Q})$, one has
 \begin{align*}
 \pi(\mathbf{Q}|X,D)=\pi(\mathbf{Q}|X)=\frac{\pi(\mathbf{Q},X)}{\pi(X)} \propto L_{t}(X;\mathbf{Q})\pi(\mathbf{Q}) \, .
 \end{align*}
 From the likelihood of a CTMC and the assumption on the prior we infer that,
 \begin{align*}
 L_{t}(X;\mathbf{Q})\pi(\mathbf{Q})\propto \prod_{i=1}^{h}\prod_{j \neq i} q_{ij}^{K_{ij}(t)}e^{-S_{i}(t)q_{ij}}\prod_{i=1}^{h}\prod_{j \neq i} q_{ij}^{\alpha_{ij}-1}e^{-\beta_{i} q_{ij}}= \prod_{i=1}^{h}\prod_{j \neq i} q_{ij}^{K_{ij}(t)+\alpha_{ij}-1}e^{-(S_{i}(t)+\beta_{i})q_{ij}} \, .
 \end{align*}
 We do not have equality here since there is no normalisation term. We generate $q_{ij}$ with $i \neq j$ from the distribution $\Gamma\big(K_{ij}(t)+\alpha_{ij}, 1/(S_{i}(t)+ \beta_{i})\big)$ (since each $q_{ij}$ is independent).


\bibliographystyle{alpha}



\begin{bibdiv}
\begin{biblist}

\end{biblist}
\end{bibdiv}


\begin{bibdiv}
\begin{biblist}

\bib{BangiaDieboldKronimusEtAl2002}{article}{
      author={Bangia, Anil},
      author={Diebold, Francis~X.},
      author={Kronimus, Andr{\'e}},
      author={Schagen, Christian},
      author={Schuermann, Til},
       title={Ratings migration and the business cycle, with application to
  credit portfolio stress testing},
        date={2002},
     journal={Journal of {B}anking \& {F}inance},
      volume={26},
      number={2},
       pages={445\ndash 474},
}

\bib{BesagGreen1993}{article}{
      author={Besag, Julian},
      author={Green, Peter~J.},
       title={Spatial statistics and {B}ayesian computation},
        date={1993},
     journal={Journal of the Royal Statistical Society. Series B
  (Methodological)},
       pages={25\ndash 37},
}

\bib{BladtEtAl2002}{article}{
      author={Bladt, Mogens},
      author={Meini, Beatrice},
      author={Neuts, Marcel~F.},
      author={Sericola, Bruno},
       title={Distributions of reward functions on continuous-time {M}arkov
  chains},
        date={2002},
     journal={Matrix-analytic methods},
       pages={39\ndash 62},
}

\bib{BrigoMaiScherer2014}{article}{
      author={Brigo, Damiano},
      author={Mai, Jan-Frederik},
      author={Scherer, Matthias~A.},
       title={Consistent iterated simulation of multi-variate default times: a
  {M}arkovian indicators characterization},
        date={2014},
     journal={Available at SSRN 2274369},
}

\bib{BladtSorensen2005}{article}{
      author={Bladt, Mogens},
      author={Sorensen, Michael},
       title={Statistical inference for discretely observed {M}arkov jump
  processes},
        date={2005},
     journal={Journal of the Royal Statistical Society: Series B (Statistical
  Methodology)},
      volume={67},
      number={3},
       pages={395\ndash 410},
}

\bib{BladtSorensen2009}{article}{
      author={Bladt, Mogens},
      author={Sorensen, Michael},
       title={Efficient estimation of transition rates between credit ratings
  from observations at discrete time points},
        date={2009},
     journal={Quantitative Finance},
      volume={9},
      number={2},
       pages={147\ndash 160},
}

\bib{Cantor2004}{article}{
      author={Cantor, Richard},
       title={An introduction to recent research on credit ratings},
        date={2004},
     journal={Journal of Banking \& Finance},
      volume={28},
      number={11},
       pages={2565\ndash 2573},
}

\bib{ContEtAl2010}{article}{
      author={Cont, Rama},
      author={Deguest, Romain},
      author={Scandolo, Giacomo},
       title={Robustness and sensitivity analysis of risk measurement
  procedures},
        date={2010},
     journal={Quantitative Finance},
      volume={10},
      number={6},
       pages={593\ndash 606},
}

\bib{ChristensenEtAl2004}{article}{
      author={Christensen, Jens~{H. E.}},
      author={Hansen, Ernst},
      author={Lando, David},
       title={Confidence sets for continuous-time rating transition
  probabilities},
        date={2004},
     journal={Journal of Banking \& Finance},
      volume={28},
      number={11},
       pages={2575\ndash 2602},
}

\bib{Culver1966}{article}{
      author={Culver, Walter~J},
       title={On the existence and uniqueness of the real logarithm of a
  matrix},
        date={1966},
     journal={Proceedings of the American Mathematical Society},
      volume={17},
      number={5},
       pages={1146\ndash 1151},
}

\bib{Cuthbert1973}{article}{
      author={Cuthbert, James~R.},
       title={The logarithm function for finite-state {M}arkov semi-groups},
        date={1973},
     journal={Journal of the London Mathematical Society},
      volume={2},
      number={3},
       pages={524\ndash 532},
}

\bib{DehayYao2007}{article}{
      author={Dehay, Dominique},
      author={Yao, Jian-Feng},
       title={On likelihood estimation for discretely observed {M}arkov jump
  processes},
        date={2007},
     journal={Australian \& New Zealand Journal of Statistics},
      volume={49},
      number={1},
       pages={93\ndash 107},
}

\bib{Fermanian2014}{article}{
      author={Fermanian, Jean-David},
       title={The limits of granularity adjustments},
        date={2014},
     journal={Journal of Banking \& Finance},
      volume={45},
       pages={9\ndash 25},
}

\bib{FrydmanSchuermann2008}{article}{
      author={Frydman, Halina},
      author={Schuermann, Til},
       title={Credit rating dynamics and {M}arkov mixture models},
        date={2008},
     journal={Journal of Banking \& Finance},
      volume={32},
      number={6},
       pages={1062\ndash 1075},
}

\bib{GelfandCarlin1993}{article}{
      author={Gelfand, Alan~E.},
      author={Carlin, Bradley~P.},
       title={Maximum-likelihood estimation for constrained-or missing-data
  models},
        date={1993},
     journal={Canadian Journal of Statistics},
      volume={21},
      number={3},
       pages={303\ndash 311},
}

\bib{GuptonEtAl1997}{book}{
      author={Gupton, Gred~M.},
      author={Finger, Christopher~Clemens},
      author={Bhatia, Mickey},
       title={Creditmetrics: technical document},
   publisher={{JP} Morgan \& Co.},
        date={1997},
}

\bib{GilksRichardsonEtAl1996}{article}{
      author={Gilks, Walter~R.},
      author={Richardson, Sylvia},
      author={Spiegelhalter, David~J.},
       title={Introducing {M}arkov chain {M}onte {C}arlo},
        date={1996},
     journal={Markov chain Monte Carlo in practice},
      volume={1},
       pages={19},
}

\bib{HobolthJensen2011}{article}{
      author={Hobolth, Asger},
      author={Jensen, Jens~Ledet},
       title={Summary statistics for endpoint-conditioned continuous-time
  markov chains},
        date={2011},
     journal={Journal of applied probability},
      volume={48},
      number={4},
       pages={911\ndash 924},
}

\bib{Inamura2006}{techreport}{
      author={Inamura, Yasunari},
       title={Estimating continuous time transition matrices from discretely
  observed data},
 institution={Citeseer},
        date={2006},
        note={(No. 06-E-7). Bank of Japan},
}

\bib{IsraelRosenthalWei2001}{article}{
      author={Israel, Robert~B.},
      author={Rosenthal, Jeffrey~S.},
      author={Wei, Jason~Z.},
       title={Finding generators for {M}arkov chains via empirical transition
  matrices, with applications to credit ratings},
        date={2001},
     journal={Mathematical finance},
      volume={11},
      number={2},
       pages={245\ndash 265},
}

\bib{JarrowLandoTurnbull1997}{article}{
      author={Jarrow, Robert~A.},
      author={Lando, David},
      author={Turnbull, Stuart~M.},
       title={A {M}arkov model for the term structure of credit risk spreads},
        date={1997},
     journal={Review of Financial studies},
      volume={10},
      number={2},
       pages={481\ndash 523},
}

\bib{Korolkiewicz2012}{article}{
      author={Korolkiewicz, Ma{\l}gorzata~Wiktoria},
       title={A dependent hidden {M}arkov model of credit quality},
        date={2012},
     journal={International Journal of Stochastic Analysis},
      volume={2012},
}

\bib{KreininSidelnikova2001}{article}{
      author={Kreinin, Alexander},
      author={Sidelnikova, Marina},
       title={Regularization algorithms for transition matrices},
        date={2001},
     journal={Algo Research Quarterly},
      volume={4},
      number={1/2},
       pages={23\ndash 40},
}

\bib{KuchlerSorensen1997}{book}{
      author={K{\"u}chler, Uwe},
      author={Sorensen, Michael},
       title={Exponential families of stochastic processes},
   publisher={Springer Science \& Business Media},
        date={1997},
      volume={3},
}

\bib{KremerWeissbach2013}{article}{
      author={Kremer, Alexander},
      author={Wei{\ss}bach, Rafael},
       title={Consistent estimation for discretely observed {M}arkov jump
  processes with an absorbing state},
        date={2013},
     journal={Statistical Papers},
      volume={54},
      number={4},
       pages={993\ndash 1007},
}

\bib{KremerWeissbach2014}{article}{
      author={Kremer, Alexander},
      author={Wei{\ss}bach, Rafael},
       title={Asymptotic normality for discretely observed {M}arkov jump
  processes with an absorbing state},
        date={2014},
     journal={Statistics \& Probability Letters},
      volume={90},
       pages={136\ndash 139},
}

\bib{Lin2011}{thesis}{
      author={Lin, Lijing},
       title={Roots of stochastic matrices and fractional matrix powers},
        type={Ph.D. Thesis},
        date={2011},
         url={https://www.escholar.manchester.ac.uk/uk-ac-man-scw:106850},
}

\bib{LongKeenanNeaguEtAl2011}{article}{
      author={Long, Kete},
      author={Keenan, Sean~C.},
      author={Neagu, Radu},
      author={Ellis, John~A.},
      author={Black, Jason~W.},
       title={The computation of optimised credit transition matrices},
        date={2011},
     journal={Journal of Risk Management in Financial Institutions},
      volume={4},
      number={4},
       pages={370\ndash 391},
}

\bib{LittleRubin2002}{book}{
      author={Little, Roderick J.~A.},
      author={Rubin, Donald~B.},
       title={Statistical analysis with missing data},
   publisher={John Wiley \& Sons},
        date={2002},
}

\bib{LandoSkodeberg2002}{article}{
      author={Lando, David},
      author={Skodeberg, Torben~M.},
       title={Analyzing rating transitions and rating drift with continuous
  observations},
        date={2002},
     journal={Journal of Banking and Finance},
      volume={26},
      number={2},
       pages={423\ndash 444},
}

\bib{McLachlanKrishnan2007}{book}{
      author={{McLachlan}, Geoffrey},
      author={Krishnan, Thriyambakam},
       title={The {EM} algorithm and extensions},
   publisher={John Wiley \& Sons},
        date={2007},
      volume={382},
}

\bib{Norris1998}{book}{
      author={Norris, James~R.},
       title={Markov chains},
   publisher={Cambridge university press},
        date={1998},
      number={2},
}

\bib{Oakes1999}{article}{
      author={Oakes, David},
       title={Direct calculation of the information matrix via the {EM}},
        date={1999},
     journal={Journal of the Royal Statistical Society: Series B (Statistical
  Methodology)},
      volume={61},
      number={2},
       pages={479\ndash 482},
}

\bib{Pfeuffer2017}{article}{
      author={Pfeuffer, Marius},
       title={{ctmcd: An R Package for Estimating the Parameters of a
  Continuous-Time Markov Chain from Discrete-Time Data}},
        date={2017},
     journal={The R Journal},
        note={To Appear},
}

\bib{PfeufferMoestelFischer2017}{techreport}{
      author={Pfeuffer, Marius},
      author={Moestel, Linda},
      author={Fischer, Matthias},
       title={{An Extended Likelihood Framework for Modeling Discretely
  Observed Credit Rating Transitions}},
 institution={University of Erlangen-Nuremberg},
        date={2017},
}

\bib{RutkowskiTarca2014}{article}{
      author={Rutkowski, Marek},
      author={Tarca, Silvio},
       title={Regulatory capital modeling for credit risk},
        date={2015},
     journal={International Journal of Theoretical and Applied Finance},
      volume={18},
      number={05},
       pages={1550034},
}

\bib{SkoglundChen2011}{article}{
      author={Skoglund, Jimmy},
      author={Chen, Wei},
       title={On the choice of liquidity horizon for incremental risk charges:
  are the incentives of banks and regulators aligned?},
        date={2011},
     journal={The Journal of Risk Model Validation},
      volume={5},
      number={3},
       pages={37\ndash 57},
}

\bib{Basel2003}{article}{
      author={Supervision, Basel Committee on~Banking},
       title={The new {B}asel capital accord},
        date={2003},
}

\bib{Basel2013}{article}{
      author={Supervision, Basel Committee on~Banking},
       title={Fundamental review of the trading book: A revised market risk
  framework},
        date={2013},
}

\bib{TsaiChan2003}{article}{
      author={Tsai, Henghsiu},
      author={Chan, KS},
       title={A note on parameter differentiation of matrix exponentials, with
  applications to continuous-time modelling},
        date={2003},
     journal={Bernoulli},
       pages={895\ndash 919},
}

\bib{TrueckOezturkmen2004}{incollection}{
      author={Tr{\"u}ck, Stefan},
      author={{\"O}zturkmen, Emrah},
       title={Estimation, adjustment and application of transition matrices in
  credit risk models},
        date={2004},
   booktitle={Handbook of computational and numerical methods in finance},
   publisher={Springer},
       pages={373\ndash 402},
}

\bib{TannerWong1987}{article}{
      author={Tanner, Martin~A.},
      author={Wong, Wing~Hung},
       title={The calculation of posterior distributions by data augmentation},
        date={1987},
     journal={Journal of the American statistical Association},
      volume={82},
      number={398},
       pages={528\ndash 540},
}

\bib{VanLoan1978}{article}{
      author={{Van}~{Loan}, C.},
       title={Computing integrals involving the matrix exponential},
        date={1978},
     journal={Automatic Control, {IEEE} Transactions on},
      volume={23},
      number={3},
       pages={395\ndash 404},
}

\bib{Wilcox1967}{article}{
      author={Wilcox, RM},
       title={Exponential operators and parameter differentiation in quantum
  physics},
        date={1967},
     journal={Journal of Mathematical Physics},
      volume={8},
      number={4},
       pages={962\ndash 982},
}

\bib{Wu1983}{article}{
      author={Wu, {C. F.}~Jeff},
       title={On the convergence properties of the {EM} algorithm},
        date={1983},
     journal={The Annals of statistics},
       pages={95\ndash 103},
}

\bib{YavinEtAl2014}{article}{
      author={Yavin, Tzahi},
      author={Wang, Eugene},
      author={Zhang, Hu},
      author={Clayton, Michael~A.},
       title={Transition probability matrix methodology for incremental risk
  charge},
        date={2014},
     journal={Journal of Financial Engineering},
      volume={1},
      number={01},
  url={http://www.worldscientific.com/doi/pdfplus/10.1142/S234576861450010X},
}

\end{biblist}
\end{bibdiv}




\end{document}